\documentclass[12 pt]{article}
\usepackage{bbm}
\usepackage{natbib}
\usepackage{amssymb}
\usepackage{amsmath}
\usepackage{amsthm}
\usepackage[onehalfspacing]{setspace}

\usepackage[colorlinks]{hyperref}
\usepackage{graphicx}

\usepackage{enumitem}
\usepackage{verbatim}
\usepackage{sectsty}
\usepackage[svgnames]{xcolor}
\usepackage{subfiles}
\usepackage[capitalize,nameinlink]{cleveref}
\usepackage[colorinlistoftodos,textsize=tiny,disable]{todonotes}
\usepackage{xcolor}

\usepackage{pgf}
\usepackage{tikz}
\usepackage{xargs}
\usepackage[font={footnotesize,it},justification=justified]{caption}
\usepackage[normalem]{ulem}
\usepackage{subfig}
\usepackage[margin=0.9in]{geometry}

\usepackage{pgfplots}
\pgfplotsset{compat=1.7}
\usepackage{booktabs}
\usepackage{soul}
\usepackage{tabularx}
\usepackage{array,multirow}
\usepackage{makecell}

\usetikzlibrary{patterns}
\usetikzlibrary{decorations.markings}

\hypersetup{
    pdftitle={Designing and Selling Hard Information},
    pdfauthor={S. Nageeb Ali, Nima Haghpanah, Xiao Lin, Ron Siegel},
    citecolor=DarkBlue,
    bookmarksnumbered=true,
    urlcolor=Indigo,linkcolor=DarkBlue
}
 \theoremstyle{definition}
 
 \theoremstyle{remark}

\theoremstyle{plain}
\newtheorem{theorem}{Theorem}
 \newtheorem{definition}{Definition}

\newtheorem{lemma}{Lemma}

\newtheorem{proposition}{Proposition}

\renewcommand{\epsilon}{\varepsilon}

\newcommandx{\nageeb}[2][1=]{\todo[linecolor=blue,backgroundcolor=blue!25,bordercolor=blue,#1]{#2}}
\newcommandx{\nima}[2][1=]{\todo[linecolor=black,backgroundcolor=black!25,bordercolor=black,#1]{#2}}
\newcommandx{\ron}[2][1=]{\todo[linecolor=red,backgroundcolor=red!25,bordercolor=red,#1]{#2}}  
\newcommandx{\xiao}[2][1=]{\todo[linecolor=green,backgroundcolor=green!25,bordercolor=green,#1]{#2}}

\newcommand{\thr}{{\threshold}}
\newcommand{\ub}{{\overline{\theta}}}
\newcommand{\lb}{{\underline{\theta}}}
\newcommand{\threshold}{\tau}
\newcommand{\testfee}{t}
\newcommand{\disclosurefee}{d}
\newcommand{\phit}{\phi_{\testfee}}
\newcommand{\phid}{\phi_{\disclosurefee}}
\renewcommand{\Re}{\mathbb{R}}
\newcommand{\tfstruct}{(T,\phi)}
\newcommand{\game}{\mathcal G\tfstruct}
\newcommand{\sosd}{\Gamma}
\newcommand{\diststruct}{(G,\phi)}

\usetikzlibrary{shapes,arrows,automata,fit}

\sectionfont{\color{DarkRed}}
\subsectionfont{\color{DarkRed}}
\subsubsectionfont{\color{DarkRed}}
\crefname{assumption}{assumption}{assumptions}
\newcommand{\appendixref}[1]{\hyperref[#1]{Appendix}}
\newcommand{\suppappendixref}[1]{\hyperref[#1]{Supplementary Appendix}}

\graphicspath{ {Figures/} } 
\allowdisplaybreaks  

\newcommand{\settheoremtag}[1]{
  \let\oldthetheorem\thetheorem
  \renewcommand{\thetheorem}{#1}
  \g@addto@macro\endtheorem{
    \addtocounter{theorem}{-1}
    \global\let\thetheorem\oldthetheorem}
  }
\usepackage{textcomp}
\begin{document}

\begin{titlepage}

\title{How to Sell Hard Information\thanks{We thank Simone Galperti, Navin Kartik, Elliot Lipnowski, Philipp Strack, Rakesh Vohra, and various seminar audiences for useful comments. The symbol \textcircled{r} indicates that the authors' names are in random order. }}
\author{
Nima Haghpanah \textcircled{r}\, S. Nageeb Ali \textcircled{r}\, Xiao Lin \textcircled{r}\, Ron Siegel \thanks{Department of Economics, Pennsylvania State University. Emails: \href{mailto:nima@psu.edu}{nima@psu.edu}, \href{mailto:nageeb@psu.edu}{nageeb@psu.edu}, \href{mailto:xiao@psu.edu}{xiao@psu.edu}, and \href{mailto:rus41@psu.edu}{rus41@psu.edu}.}
}

\maketitle

\begin{abstract}
The seller of an asset has the option to buy hard information about the value of the asset from an intermediary. The seller can then disclose the acquired information before selling the asset in a competitive market. We study how the intermediary designs and sells hard information to robustly maximize her revenue across all equilibria. Even though the intermediary could use an accurate test that reveals the asset's value, we show that robust revenue maximization leads to a noisy test with a continuum of possible scores that are distributed exponentially. In addition, the intermediary always charges the seller for disclosing the test score to the market, but not necessarily for running the test. This enables the intermediary to robustly appropriate a significant share of the surplus resulting from the asset sale even though the information generated by the test provides no social value.   

\nima{This is Nima.}
\nageeb{This is Nageeb.}
\ron{This is Ron.}
\xiao{This is Xiao.}

\end{abstract}

\thispagestyle{empty} 

\end{titlepage}
\clearpage
	\setcounter{tocdepth}{2}
	\tableofcontents
	\thispagestyle{empty}
\clearpage
\newpage
\setcounter{tocdepth}{1}
\listoftodos[Blue is Nageeb, Black is Nima, Red is Ron, and Green is Xiao. Orange is General.]

\newpage
\setcounter{page}{1}
\nageeb{Link to Zoom Room: https://psu.zoom.us/j/93590047769?pwd=YkhNdE9neFhVNVRuOHJoUUF4cUtlQT09}

\section{Introduction}\label{Section-Introduction}

This paper studies settings in which individuals purchase hard information from an intermediary that they can verifiably disclose to influence the actions of others.\footnote{For a survey of the literature on the design and sale of information by intermediaries see \cite{bergemann2019markets}.} Such settings are ubiquitous. 
For example, entrepreneurs often seek evidence that they can disclose to venture capitalists to obtain more funding, sellers of physical and financial assets routinely pay for evaluations that enable them to obtain better prices, and workers commonly seek certification before applying for positions. The size of the market for hard information, and certification in particular, is in the hundreds of billions.\footnote{See, for example, \href{https://tinyurl.com/y5ce9p4b}{https://tinyurl.com/y5ce9p4b}.} 

An important justification for the existence of information intermediaries in these settings is that they generate economic value. The information they provide may facilitate assortative matching or alleviate moral hazard or adverse selection.\footnote{Issues of moral hazard and certification have been studied by \cite{albano2001strategic}, \cite{marinovic:2018}, \cite{Zapechelnyuk:2020}, and \cite{saeedi2020optimal} among others.} It may also affect how the surplus is divided between the parties (other than the intermediary). For example, a seller who knows that his asset is valuable may be unable to credibly convey that information to buyers on his own, but objective evidence obtained from a trusted third-party would allow him to do so, enabling him to negotiate a better price for his asset. 

But would we expect to find information intermediaries where they do not provide economic value? Perhaps not, since purchasing information from them is voluntary and if they do not provide economic value, they can make a profit only by reducing the surplus of the other agents. We consider such a setting with an information intermediary and show, perhaps surprisingly, that the intermediary is able to appropriate a significant part of the surplus arising from trade by generating noisy information and charging certain fees, even when the equilibrium played by the other parties is chosen adversarially to the intermediary's interests.

In our model, an agent owns an asset that he would like to sell in a competitive market. Both the agent and the market have symmetric information about the asset's market value. Before selling the asset, the agent can purchase additional, hard information from an intermediary about the asset's value that the agent can share with the market to improve the terms of trade. The intermediary chooses a \emph{test}, which stochastically maps the asset's value to a score that can be verifiably disclosed, and a two-part tariff for her services. The tariff comprises a \emph{testing fee} for running the test and a \emph{disclosure fee} for disclosing the resulting score to the market. If the agent pays the testing fee, the test is run and the agent observes the resulting score. He then chooses whether to pay the disclosure fee to obtain hard information that enables him to disclose the score to the market. The market cannot distinguish between the agent not disclosing the score and the test not having been run. Because the market is competitive, the market price for the asset following disclosure or non-disclosure equals the asset's expected value conditioning on all the information available to the market and the equilibrium choices of the agent. 

The intermediary's choice of test and fees shapes the disclosure game between the agent and the market: the test determines the degree and kind of private information the agent will have in the trading relationship, and the fees determine his learning and disclosure costs. The intermediary's goal is to extract as much surplus from the agent as possible. Should the intermediary use a noisy test to determine the asset's value or should she use an accurate one? What combination of testing and disclosure fees should she charge? Finally, how much of the trading surplus can she extract? 

To address these questions, we first observe that the agent's outside option depends on the market's expectation of the agent's behavior. If the market expects the agent to pay the testing fee with some probability but the agent attempts to opt out of the disclosure game by not paying the testing fee, he cannot prove to the market that he chose to opt out. Instead, when the market sees no score being disclosed, it rationally concludes that the test may have been run but the agent chose not to disclose a low score. The market weights this ``bad news'' contingency and the ``no news'' contingency in which the agent decided not to pay the testing fee, where the weights depend on what the market believes the agent does in equilibrium. Thus, unlike in a standard mechanism-design problem, the agent's outside option depends on both the test-fee structure and the equilibrium that is played in the induced game between the agent and the market.

If the intermediary can select the equilibrium of the induced game, then an optimal test-fee structure comprises a fully revealing test, a high testing fee, and no disclosure fee. The intermediary selects an equilibrium in which the agent always pays the testing fee, and when the agent chooses not to disclose the score, the market believes that the test revealed that the asset has its lowest possible value (say $\lb$). The agent discloses the test score whenever it reveals that the asset's value is not the lowest possible value. The high testing fee charged by the intermediary extracts all of the agent's expected surplus (minus $\lb$) from selling the asset, so his expected payoff is $\lb$.\footnote{The intermediary can also extract all the surplus by using a binary score test, making testing free, and charging a high disclosure fee.} This is consistent with a key intuition from standard mechanism design: since the agent's payoff beyond what is needed to satisfy his individual rationality constraint is due to information rents, when the agent starts with no private information the designer can keep the agent's payoff at his individual rationality level, extract the full surplus, and achieve this by charging an upfront fee.

But the game induced by this test-fee structure has another equilibrium, in which the agent never pays the testing fee so the intermediary's revenue is zero. In this equilibrium, the market treats non-disclosure as ``no news,'' and the resulting market price is the asset's ex ante expected value. Given this, it is optimal for the agent not to pay for the intermediary's services. Thus, choosing this test-fee structure leaves the intermediary vulnerable to obtaining zero revenue. We show in \cref{prop-whyrobust} that this is not an accident: any test-fee structure that has an equilibrium in which the intermediary extracts (approximately) all the surplus also has an equilibrium in which the intermediary's revenue is (approximately) zero. 

Motivated by the above discussion, our goal is to identify \emph{robustly optimal test-fee structures}, namely those that guarantee the highest revenue to the intermediary across \emph{all} equilibria of the induced game. This corresponds to the intermediary choosing the test-fee structure that maximizes her revenue assuming that the equilibrium of the induced game is selected adversarially to her interests. Our motivation for studying robustly optimal test-fee structures is twofold. First, the intermediary may be unable to coordinate the behavior of the agent and the market on her most preferred equilibrium. The uncertainty about which equilibrium will be played could motivate her to be cautious and therefore use test-fee structures that guarantee her a high revenue across all equilibria.\footnote{Our focus on adversarial equilibrium selection is shared by a rapidly growing literature in mechanism and information design, including \cite*{bergemann2017first}, \cite{du2018robust}, \cite*{hoshino2019multi}, \citet*{ziegler2019adversarial}, \cite{inostroza2017persuasion}, \cite*{DworczakPavan2020}, \cite*{halac2020raising}, \cite*{halac2020rank}, and \cite*{mathevet2020information}.} Second, for any test-fee structure, the sum of the agent and the intermediary's revenue is constant across equilibria and equal to the asset's ex ante expected value, so the intermediary's least preferred equilibrium is the agent's most preferred equilibrium. And the agent and the market may be able to coordinate on the agent's preferred equilibrium. As we will see, the intermediary obtains a substantial share of the surplus when using the robustly optimal test-fee structure. This will show that the ability of the intermediary to secure a profit even though the information she provides has no social value and does not benefit the agent (or the market) ex-ante does not depend on the equilibrium played in the induced game.

Finding the robustly optimal test-fee structure involves optimizing across all test-fee structures and all equilibria of the induced games. Different tests induce different distributions of private information about the asset's value for the agent in the induced games, and different fees change the agent's incentives to obtain and disclose this private information. Thus, the optimization entails comparing the equilibria of disclosure games that vary in both the amount of the agent's private information and his disclosure costs. Despite this richness, we find that robustly optimal test-fee structures take a relatively simple form regardless of the distribution of the asset's value. As we illustrate in \autoref{figure-intro}, robustly optimal tests are in the ``step-exponential-step'' class: they feature an exponential distribution over a continuum of scores (even if the asset's value is drawn from a finite set), one atom below this continuum, and possibly one atom above it. The optimal disclosure fee is always positive, but the optimal testing fee may be positive or $0$. The resulting payoff to the intermediary is positive and bounded away from the full surplus.

\begin{figure}[t]
	\centering

	\begin{tikzpicture}[domain=2:4, scale=3.7, ultra thick,decoration={
		markings,
		mark=at position 0.5 with {\arrow{>}}}]

    \draw[dotted] (2,1) node[left]{1} -- (3,1) -- (3,0) node[below]{$1$};

	\draw[<->] (2,1.2) node[left]{$G(s)$}-- (2,0) -- (3.2,0) node[below]{$s$};
	\draw (2,.05) -- (2,0) node[below]{$0$};

	\draw[line width=3,domain=2.3:2.7,smooth,variable=\x] (2,0) -- (2.1,0)  -- (2.1,.2) -- (2.3,.2) -- plot ({\x},{0.2*(e^((\x-2.3)/0.4))}) -- (2.7,.5436) -- (2.85,.5436) -- (2.85,1) -- (3,1);

	\end{tikzpicture}
	\caption{A robustly optimal score distribution if the asset quality $\theta$ is drawn from $[0,1]$, where $G$ is the marginal CDF on scores. The score distribution features atoms on a low and a high score, and an exponential distribution over a continuum of intermediate scores.}
     \label{figure-intro}
\end{figure}
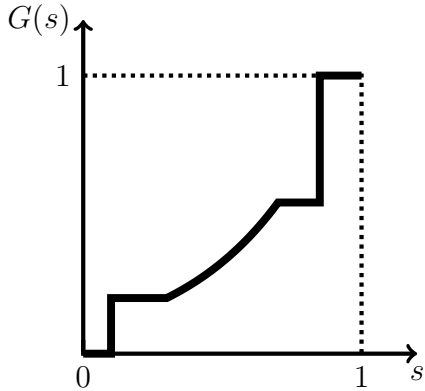

To derive these features of the robustly optimal test-fee structure we first observe that the intermediary does not provide any added value to the agent ex ante. This is because the market draws correct (Bayesian) inferences and is competitive, so for every test-fee structure and any equilibrium, the ex ante expected market price is the ex ante expected value of the asset. Thus, for every test and for all positive fees, the agent strictly prefers an equilibrium in which the market expects him not to have the asset tested and consequently offers him the asset's ex ante expected value. The agent then retains the full surplus and the intermediary's revenue is zero. If the market were to observe whether the asset is tested, the agent could achieve this as an equilibrium outcome by not paying the testing fee.  But the market \emph{does not} observe whether the asset is tested. Therefore, the market's expectation of whether the asset has been tested and the agent observed the resulting score must be consistent with the agent's unobserved equilibrium choice of whether to pay the testing fee.

This is how the intermediary obtains a positive payoff robustly: she uses \emph{option value as a carrot} to make it non-credible for the agent not to pay for the asset to be tested. 
Because the market only learns that the test has been run if the agent pays the disclosure fee and discloses the test score, the intermediary creates option value for the agent by offering a test that generates high test scores with some probability and setting the testing and disclosure fees sufficiently low. If this option value is sufficiently high, the agent cannot credibly refrain from paying the testing fee. The intermediary then obtains at least the testing fee in every equilibrium. Moreover, the market then treats non-disclosure as concealing a low score, which further motivates the agent to pay the disclosure fee and disclose the test score. In effect, the agent is trapped by market expectations that he has paid the testing fee and will disclose if the test score is sufficiently high.

But even if the agent pays the testing fee with certainty, multiple equilibria may exist. These equilibria differ in the set of scores that the agent discloses, and therefore in the probability of disclosure and the intermediary's revenue. The exponential test score distribution is robustly optimal because it eliminates potential equilibria in which the agent discloses with low probability. We develop an intuition for this result by showing that the intermediary can be thought of as choosing an optimal ``demand curve for testing,'' subject to the demand curve being feasible and the quantity of testing demanded corresponding to the one in the equilibrium least favorable to the intermediary. 
We illustrate this approach in \cref{Section-Example}, and provide a general analysis in \Cref{Section-Model,Section-SimplifyingProblem,Section-RobustlyOptimalTestFeeStructures}. 

\Cref{Section-Extension} describes two extensions. First, we consider a setting in which testing is costly for the intermediary. We show that if testing costs increase in the Blackwell order, then our main results continue to hold, that is, there exists a robustly optimal test in the step-exponential-step class. Moreover, if the increase is strict, then every robustly optimal test is in this class. Second, we consider an intermediary who can sell the agent multiple pieces of evidence, and gives him the choice of which to disclose. We show that this additional flexibility does not improve the intermediary's revenue guarantee.

The main takeaway of our paper is that even if information is neither socially valuable nor ex ante valuable to a seller, an intermediary selling hard information can profit from the seller's inability to commit. The intermediary robustly maximizes her revenue by using test-fee structures that generate option value, use noise, and include strictly positive disclosure fees. Thus, the presence of profitable intermediaries, as in the large certification industry is not, in itself, evidence that the provision of hard information improves the welfare of market participants.

Our work builds on the rich literature on verifiable disclosure and persuasion games, initiated by \cite{grossman1981informational} and \cite{milgrom1981good}. Their main insight is that if a privately informed agent could costlessly and verifiably disclose evidence about his type, the unique equilibrium involves full disclosure. The subsequent literature suggests a number of mechanisms that dampen this force, including exogenously costly disclosure \citep{jovanovic1982truthful,verrecchia1983discretionary} and lacking evidence with positive probability \citep{dye1985strategic}. \cite{matthews1985quality} and \cite{shavell1994acquisition} consider an uninformed agent who decides whether to take a fully revealing test. \cite{matthews1985quality} show that the unique equilibrium involves testing and full disclosure of the test result when disclosure is voluntary but involves no testing when disclosure of the test result is mandatory. \cite{shavell1994acquisition} assumes the agent bears a privately-known cost of testing, and studies how this cost dampens unraveling. 

The features in the models above also apply to our setting: the agent faces a cost of obtaining information about her type and a cost of disclosing that information in a verifiable form, and with positive probability, the agent may lack evidence. But we derive these features endogenously because the intermediary chooses the evidence structure as well as the cost of learning and disclosing the evidence; the probability that the market attributes to the agent having evidence is determined in equilibrium. Treating these features as endogenous objects reveals a tradeoff: all else equal, the intermediary would like the market to unravel (so that the agent discloses with maximal probability), but the instruments from which she earns revenue are exactly those that counter unraveling. This ``quantity-price'' tradeoff leads to the price-theoretic approach to evidence generation we develop, in which the intermediary both chooses the optimal price and designs the optimal demand curve subject to constraints that correspond to Bayes rule and adversarial equilibrium selection.

A closely related strand of the disclosure literature studies choices made by agents to influence market perceptions. \cite{benporath2018disclosure} model how an agent chooses projects when he obtains evidence of project returns with positive probability. Their analysis emphasizes option value from the possibility of disclosure as motivating the agent to choose riskier projects. \cite{demarzo2019test} study how an agent chooses tests and disclosures to influence the market valuation of his asset. While most of their paper concerns evidence being costlessly generated in-house by the agent, they also consider the choices that would be made by a monopolistic intermediary who can charge a testing fee but not a disclosure fee, assuming that equilibria are selected to favor the intermediary. Our analysis shows that when the intermediary chooses to protect herself from the worst equilibrium (or assumes that the agent and market can coordinate on the agent's best equilibrium) and can charge a disclosure fee, she optimally designs her tests differently.\footnote{Another strand of the literature less closely related to our work studies organizational settings in which an agent exerts effort to acquire hard information to persuade others; see \cite{aghion1997formal}, \cite{dewatripont1999advocates}, and \cite{che2009opinions}. \cite{shishkin2019evidence} considers a sender who commits to an evidence structure and then chooses whether to disclose evidence to persuade a receiver. Both the probability of obtaining evidence and the form that evidence takes are determined endogenously. He finds that the optimal evidence structure is binary certification. } 

In the papers discussed above, and in our work, the players are symmetrically informed at the outset of their interaction. By contrast, \cite{lizzeri1999information} studies signaling dynamics when the agent is perfectly and privately informed at the outset about the value of the asset and chooses whether to have the asset tested. In his model, testing is voluntary (and observable) and if the agent has the asset tested, disclosure is mandatory. The intermediary's gains come from the prejudicial expectation that the market forms when the privately informed agent chooses to not have the asset tested. \citeauthor{lizzeri1999information} shows that the intermediary can extract the full surplus using a nearly uninformative test.\footnote{\cite{harbaugh/rasmusen:2018} study a related model in which the intermediary's objective is to provide as much information as possible to the receiver. They show that the intermediary nevertheless chooses to provide coarse information to encourage the privately informed agent to have the asset tested.} Were the agent instead uninformed prior to testing, the agent would retain the full surplus and the intermediary's revenue would be $0$.\footnote{As an application of a novel result on Bayesian updating, \cite{kartik2020information} study a version of \citeauthor{lizzeri1999information}'s model in which the agent is privately and imperfectly informed about the asset's value.} Our work sheds light on a different and complementary economic force: we show that when \emph{both} testing and disclosure are voluntary, an intermediary can gain from using option-value to incentivize an uninformed agent to always have the asset tested. The intermediary then capitalizes on prejudicial inferences that are made if the agent chooses not to disclose the test score.

\section{An Example}\label{Section-Example}

Consider an agent who will sell an asset in a competitive asset market. The market value of the asset, $\theta$, is either $0$ or $1$, each with equal probability. Neither the agent nor the asset market know the asset's value, but an intermediary can run a test that generates information about the asset's value. The intermediary chooses the test, $T$, which stochastically maps the asset's value to an unbiased score $s$ in $[0,1]$, so $s=E[\theta|s]$. If the agent wants the asset tested, he has to pay the intermediary a testing fee $\phit$. If the asset is tested, the score is reported to the agent. The agent then chooses whether to pay an additional disclosure fee $\phid$ to disclose the score as hard information to the market; otherwise, no score is disclosed to the market. If no score is disclosed, the market cannot distinguish between (i) the asset not being tested and (ii) the asset being tested but the agent not disclosing the score. The intermediary chooses both the test $T$ and the fees (a test-fee structure), and her choice induces a disclosure game between the agent and the asset market. Some test-fee structures have multiple equilibria, and the intermediary's objective is to choose a test-fee structure that guarantees the highest payoff across equilibria. 

We use this example to illustrate several features of our analysis. We show that the intermediary benefits from using noisy tests that pool low and high asset values as intermediate scores. We visually depict how our problem maps to setting a price on an optimally chosen demand curve. And we show why the robustly optimal test uses an exponential score distribution. To illustrate each of these features, it suffices to assume that testing is free and the intermediary charges only a disclosure fee. As we later show, these features also arise when it is optimal to charge a testing fee.

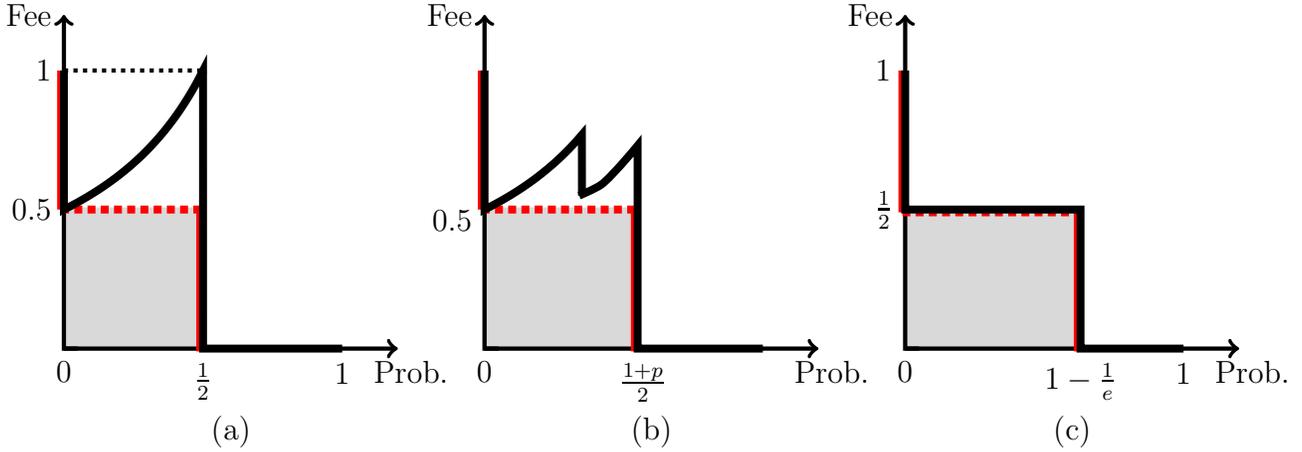
\begin{figure}[t]
	\centering
	
    \setlength{\tabcolsep}{-8pt}
	\begin{tabular}{ccc}

	\begin{tikzpicture}[domain=0:4, scale=3.7, ultra thick,decoration={
		markings,
		mark=at position 0.5 with {\arrow{>}}}]    
	\filldraw[black!15] (0.5,0) -- (.5,.5) -- (0,.5) -- (0,0) --cycle;
	
	\draw[<->] (1.2,0) node[below,xshift=5]{Prob.}-- (0,0) -- (0,1.2) node[left]{Fee};
	\draw (.05,0) -- (0,0) node[below]{0};

    \draw[line width = 3, red] (-0.01,1) -- (-0.01,0.5);
    
    \draw[line width = 3, red,dotted] (-0.01,0.5) -- (.49,0.5);
    
    \draw[line width = 3, red] (.49,0.5)-- (.49,-0.01);
	
	\draw[line width=3,domain=0:0.5,smooth,variable=\x] (0,1) -- (0,0.5) -- plot ({\x},{1/(2-2*\x)}) -- (0.5,0) node[below]{$\frac{1}{2}$} -- (1,0) node[below]{1};
    
    \draw[dotted] (0,1) node[left]{1} -- (.5,1);
    
    \draw (0,0.5) node[left]{$0.5$};

	\node at (0.6,-0.3) {(a)};
	
	\end{tikzpicture}
	&
		\begin{tikzpicture}[domain=0:4, scale=3.7, ultra thick,decoration={
		markings,
		mark=at position 0.5 with {\arrow{>}}}]    
	\filldraw[black!15] (0.55,0) -- (.55,.5) -- (0,.5) -- (0,0) --cycle;
	
	\draw[<->] (1.2,0) node[below,xshift=5]{Prob.}-- (0,0) -- (0,1.2) node[left]{Fee};
	\draw (.05,0) -- (0,0) node[below]{0};
	
	\draw[line width = 3, red] (-0.01,1) -- (-0.01,0.5);
	
	\draw[line width = 3, red,dotted] (-0.01,0.5) -- (.54,0.5);
	
	\draw[line width = 3, red] (.54,0.5)-- (.54,-0.01);
	
	\draw[line width=3,domain=0:0.35,smooth,variable=\x] (0,1) -- (0,0.5) -- plot ({\x},{1/(2-2*\x)}) -- (0.35,.55);
	
	\draw[line width = 3] plot [smooth] coordinates {(.34,.55) (.42,.59) (.49,.66) (.55,.73)} -- (.55,0) node[below,xshift=2]{$\frac{1+p}{2}$}-- (1,0);

	\draw (0,0.5) node[left,yshift=-4]{$0.5$};

	\node at (0.6,-0.3) {(b)};
	
	\end{tikzpicture}
	&
		\begin{tikzpicture}[domain=0:4, scale=3.7, ultra thick,decoration={
		markings,
		mark=at position 0.5 with {\arrow{>}}}]    
	\filldraw[black!15] (0.63,0) -- (.63,.5) -- (0,.5) -- (0,0) --cycle;	
	
	\draw[<->] (1.2,0) node[below,xshift=5]{Prob.}-- (0,0) -- (0,1.2) node[left]{Fee};
	\draw (.05,0) -- (0,0) node[below]{0};
	
	\draw[line width=3,red] (-0.01,1) -- (-0.01,0.49);

	\draw[line width=3,red,dotted]  (-0.01,0.49) -- (.62,.49);
	
	\draw[line width=3,red] (.62,.49) -- (.62,-0.01);

	\draw[line width=3,domain=0:0.5,smooth,variable=\x] (0,1) node[left]{1} -- (0,0.5) -- (.63,.5) -- (.63,0) node[below]{$1-\frac{1}{e}$} -- (1,0) node[below]{1};
    
    \draw (0,0.5) node[left]{$\frac{1}{2}$};

	\node at (0.6,-0.3) {(c)};
	
	\end{tikzpicture}
	\end{tabular}
	\vspace{-0.16in}
	\caption{The black curve graphs the demand correspondence, and the red curve represents the robust demand curve for (a) the fully revealing test, (b) the test with three scores, and (c) the robustly optimal test.}
	\label{figure-examplesection} 
\end{figure}

We first describe the intermediary's revenue guarantee if she used a fully revealing test, i.e., a test in which the score is equal to the asset's value $\theta$. 
\autoref{figure-examplesection}(a) depicts our analysis of this case using a \emph{(inverse) demand correspondence}, which for any disclosure fee in $[0,1]$ (on the vertical axis) traces the corresponding probabilities of disclosure consistent with equilibria of the induced game. 
For any such disclosure fee, there is an equilibrium in which the agent has the asset tested, discloses a score of $1$, and conceals a score of $0$; if no score is disclosed, the market offers a price of $0$ for the asset. The agent has no incentive to deviate because the payoff of $1 - \phid$ from disclosing a score of $1$ exceeds the payoff from not disclosing, which is in turn equal to that from disclosing a score of $0$. This is the unique equilibrium for any disclosure fee in $(0,1/2)$, and the resulting revenue for the intermediary is half the disclosure fee.\footnote{If the disclosure fee is $0$, then for any probability $p$, there is an equilibrium in which the agent discloses a score of $1$ and with probability $p$ discloses a score of $0$, since the agent's payoff is then $0$ whether or not he discloses.} 

For disclosure fees that weakly exceed $1/2$, other equilibria exist. There is an equilibrium in which the agent has the asset tested but never discloses the score; if no score is disclosed, the market offers a price of $1/2$. Given this market price, the agent prefers not to disclose a score of $1$ (or $0$) because the disclosure fee is at least $1/2$. Because the agent does not disclose any score, the intermediary's revenue is zero in this equilibrium. There is also a mixed strategy equilibrium in which the agent discloses a score of $1$ with an interior probability. For disclosure fees higher than $1/2$, the demand correspondence in \autoref{figure-examplesection}(a) depicts the three ``quantities of disclosure'' associated with these three equilibria.

Since we study the intermediary's robust revenue across equilibria, we identify the ``robust demand curve" for disclosure, which maps each disclosure fee (price) to the \emph{lowest} probability of disclosure (quantity) across all the equilibria associated with that fee. This is the red curve in \autoref{figure-examplesection}(a). The robust revenue for a given disclosure fee is the product of the disclosure fee and the lowest equilibrium probability of disclosure, i.e., the area of the rectangle under the robust demand curve at the point associated with that disclosure fee. The maximal robust revenue of $\approx 1/4$ for the fully revealing test is the area of the shaded rectangle in \autoref{figure-examplesection}(a), obtained from a disclosure fee that is slightly below $1/2$.\footnote{As is common with adversarial equilibrium selection, the intermediary cannot robustly achieve a payoff of exactly $1/4$ with a disclosure fee of $1/2$---since there is then an equilibrium in which she obtains a payoff of $0$.}

Let us now see how the intermediary can robustly improve her revenue by introducing an intermediate score that pools the asset's two possible values, $0$ and $1$. Consider a test with three possible scores---$0$, $3/4$, and $1$---generated by the conditional distribution in \cref{table:1.1}. The demand correspondence and the robust demand curve for this test are shown in \autoref{figure-examplesection}(b). For strictly positive disclosure fees less than $1/2$, there is a unique equilibrium: the agent has the asset tested, and discloses his score if it is $3/4$ or $1$; if no score is disclosed, the market offers a price of $0$. The probability of disclosure in this equilibrium is $(1+p)/2$. For disclosure fees higher than $1/2$, other equilibria exist, and all of them feature a lower probability of disclosure. Thus, by setting a disclosure fee slightly below $1/2$, the intermediary can obtain a revenue of $(1+p)/4$, which is more than her revenue from the robustly optimal fully revealing test. By introducing the intermediate score $3/4$, the intermediary increases the quantity of disclosure demanded because the agent pays the disclosure fee not only when the asset's value is $1$ but also, with some probability, when its value is $0$.

\begin{table}[t]
\centering
\makebox[0pt][c]{\parbox{1\textwidth}{\centering
        \begin{tabular}{  |c | c |c|c|}
           \hline
           &  $s=0$ & $s=3/4$ & $s=1$\\ \hline
           ${\theta=0}$ & $1-p$ & $p$ &$0$\\ \hline
          $\theta=1$ & $0$ & $3p$& $1-3p$\\ \hline
        \end{tabular}
        \caption{Distribution of score $s$ conditional on asset value $\theta$, where $p$ is in $(0,1/3)$.}
        \label{table:1.1}
}}
\end{table}

But the intermediary is limited in how much pooling she can introduce: a value of $p$ that is too high introduces a bad equilibrium in which the agent has the asset tested but only discloses a score of $1$. If no score is disclosed, the market concludes that the score is either $0$ or $3/4$ and offers a price of $3p/(1+3p)$. In this equilibrium, the probability of disclosure is $(1-3p)/2$. To eliminate this equilibrium, the intermediary sets the value of $p$ so that if the score is $3/4$, the agent prefers to disclose his score rather than obtain the non-disclosure price:
\begin{align*}
    \underbrace{\frac{3}{4}-\phid}_{\text{Disclose a score of $3/4$}}>\underbrace{\frac{3p}{1+3p}}_{\text{Nondisclosure price}}
\end{align*}
With a disclosure fee slightly below $1/2$, the above inequality holds for $p$ less than $1/9$. The robust revenue is attained by setting $p=1/9$ and charging a disclosure fee slightly below $1/2$. 
The restriction on $p$ is a ``cross-equilibrium'' constraint, which ensures that a low-revenue equilibrium is not created. Cross-equilibrium constraints lead to the exponential score distribution in the robustly optimal tests. 

To see why, let us return to the demand curve approach. These constraints lead to tests that induce a rectangular demand correspondence like the one in \autoref{figure-examplesection}(c): intuitively, whenever the demand correspondence or the robust demand curve are not a rectangle, there is some slack that the intermediary can use to change the test in a way that increases the probability of disclosure (by having more pooling) in the equilibrium that has the lowest probability of disclosure, , i.e., without violating cross-equilibrium constraints.

Which tests lead to rectangular demand correspondences? Those whose marginal score distribution includes an exponentially distributed component. To see why, notice that with a rectangular demand correspondence, there is a disclosure fee $\phid$ for which there is an interval of equilibrium disclosure probabilities. In \autoref{figure-examplesection}(c), this disclosure fee is $\phid=1/2$. This continuum of equilibria implies that there is a continuum of scores $s$ for which 
\begin{align}\label{Equation-ExamplesEqmThreshold}
    s-\phid=E[s'|s'\leq s].
\end{align}
The LHS above is the market price from disclosing a score $s$, and the RHS is that from not disclosing when the set of scores that do not disclose are those weakly below $s$. Thus, when \eqref{Equation-ExamplesEqmThreshold} holds there is an equilibrium with a disclosure threshold of $s$. Using $G(s)$ as the marginal CDF on scores and $G(s'|s'\leq s)$ as the conditional CDF on scores no higher than $s$, the RHS can be re-written as $\int_0^s (1-G(s'|s'\leq s))ds'$. It follows then that \autoref{Equation-ExamplesEqmThreshold} is equivalent to
\begin{align}\label{Equation-ExamplesDiff}
    \phid&=\int_0^s G(s'|s'\leq s)ds'=\frac{\int_0^{s} G(s')ds'}{G(s)}= \bigg(\frac{d}{d s}\Big(\ln(\int_0^{s} G(s')ds')\Big)\bigg)^{-1}. 
\end{align} 
Because the above is true for an interval of scores $s$, it defines a differential equation whose solution is $G(s)= \alpha e^{\frac{s}{\phid}}$ for some constant $\alpha$, , i.e., the exponential score distribution.\footnote{The exponential score distribution may bring to mind the work of \cite{roesler2017buyer} and \cite{condorelli2020information}. In those papers and in our work, a player optimally chooses a demand curve by manipulating an information structure or a distribution of values. In those papers, a buyer optimally chooses a unit-elastic demand curve that leads a seller to charge a low price and be indifferent between that price and higher prices. In a similar vein, \cite{ortner2018making} show that a principal can reduce the cost of monitoring by paying a monitor wages that generate a unit-elastic demand curve for bribes. In our paper, adversarial equilibrium selection leads the intermediary to choose a test that generates a rectangular demand curve, which corresponds to an exponential distribution of scores. The driving force is not to generate any indifference but to eliminate equilibria in which the agent does not disclose intermediate scores.}

\begin{figure}[t]
	\centering
		\setlength{\tabcolsep}{0 pt}
	\begin{tabular}{ccc}
	\begin{tikzpicture}[domain=0:3, scale=3, ultra thick,decoration={
		markings,
		mark=at position 0.5 with {\arrow{>}}}]    
	
	\draw[dotted] (0,1) node[left]{1} -- (1,1) -- (1,0) node[below]{$1$};

	\draw[<->] (0,1.2) node[left]{$G(s)$}-- (0,0) -- (1.2,0) node[below]{$s$};
	\draw (0,.05) -- (0,0) node[below]{$0$};
	\draw (0.57,0.6) node[right,rotate = 45]{$e^{2(x-1)}$};
	
	\draw[line width=3,domain=.5:1,smooth,variable=\x] plot ({\x},{(e^(2*(\x-1)))});
	
	\draw[loosely dotted] (.5,.3678) -- (.5,0) node[below]{$0.5$};
	
	\draw[line width = 3] (.5,.3678) -- (0,.3678) node[left]{$1/e$};
	
	\draw (0.5,-0.2) node[below]{(a)}   ;

	\end{tikzpicture}
	&
	\begin{tikzpicture}[domain=0:3, scale=3, ultra thick,decoration={
		markings,
		mark=at position 0.5 with {\arrow{>}}}]    
	
	\draw[dotted] (0,1) node[left]{1} -- (1,1) -- (1,0) node[below]{1};

	\draw[<->] (0,1.2) node[left]{$G(s|0)$}-- (0,0) -- (1.2,0) node[below]{$s$};
	\draw (0,.05) -- (0,0) node[below]{0};
	
	\draw[line width=3,domain=.49:1,smooth,variable=\x] plot ({\x},{(3-2*\x)*e^(2*(\x-1))});
	
	\draw[loosely dotted] (.5,.7356) -- (.5,0) node[below]{$0.5$};
	
	\draw[line width = 3] (.5,.7356) -- (0,.7356) node[left]{$2/e$};

		\draw (0.5,-0.2) node[below]{(b)}   ;
	\end{tikzpicture}
	&
	\begin{tikzpicture}[domain=0:3, scale=3, ultra thick,decoration={
		markings,
		mark=at position 0.5 with {\arrow{>}}}]    
	
	\draw[dotted] (0,1) node[left]{1} -- (1,1) -- (1,0) node[below]{1};

	\draw[<->] (0,1.2) node[left]{$G(s |1)$}-- (0,0) -- (1.2,0) node[below]{$s$};
	\draw (0,.05) -- (0,0) node[below]{0};
	
	\draw[line width=3,domain=.49:1,smooth,variable=\x] plot ({\x},{((2*\x-1)*e^(2*(\x-1)))});
	
	\draw (0.5,0) node[below]{$0.5$};
	
	\draw[line width = 3] (.5,0) -- (0,0);
	
	\draw (0.5,-0.2) node[below]{(c)}   ;
	\end{tikzpicture}
	
	\end{tabular}
	\caption{\footnotesize The robustly optimal distribution of scores where (a) $G(\cdot)$ is the marginal probability, (b) $G(\cdot|0)$ is the conditional probability of scores for $\theta=0$, and (c) $G(\cdot|1)$ is the conditional probability for $\theta=1$.}
	\label{figure-example} 
\end{figure}
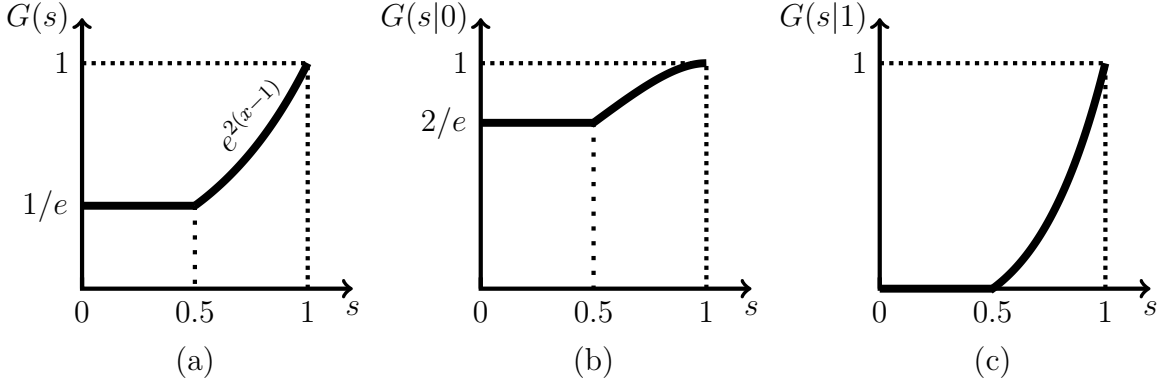

We illustrate this solution in \autoref{figure-example}, where (a) shows the marginal score distribution in the robustly optimal test. This marginal distribution is generated by a test in which if the value of the asset is $0$, with probability $2/e$ the test score is $0$ and with complementary probability the test score is distributed exponentially on $[1/2,1]$, and if the value of the asset is $1$, the test score is distributed exponentially on $[1/2,1]$. The testing fee is $0$ (by assumption) and the disclosure fee is slightly less than $1/2$. This test-fee structure induces a game with a unique equilibrium: the agent has the asset tested, and discloses the score if it exceeds the disclosure fee. If no score is disclosed the market offers a price of $0$. The resulting revenue for the intermediary is $\approx \frac{1}{2}\left(1-\frac{1}{e}\right)$, which is more than $1/4$ but less than the full surplus of $1/2$.

\section{Model}\label{Section-Model}
\subsection{The Setting}

A risk-neutral agent sells an asset in a competitive market comprising two risk-neutral buyers who have a common value for the asset. The agent and both buyers are symmetrically informed about the asset's market value, $\theta$, knowing only that it is drawn according to a distribution $F$ with support $\Theta\subseteq [\underline{\theta},\overline{\theta}]$, where $0\leq\underline\theta<\overline\theta<\infty$ are the lowest and highest values in $\Theta$, with an expected value of $\mu$. 

Prior to selling the asset, the agent can pay an intermediary to evaluate the asset, and if he does so, he has the option, for a fee, to disclose that evaluation to the market. The intermediary chooses how to evaluate the asset and sets the fees for her services, and we call such a scheme a \emph{test-fee structure}. A test-fee structure comprises:  
\begin{enumerate}[noitemsep,label=\alph*.]
   \item a \emph{test}, $T$, where $T:\Theta\rightarrow \Delta S$ is a measurable function that stochastically maps the asset's value to a score in some set $S$,
    \item a \emph{testing fee}, $\phit\in \Re$, which the agent pays for the asset to be evaluated, and
    \item a \emph{disclosure fee}, $\phid\in \Re$, which the agent pays to disclose the score to the market. 
\end{enumerate}

Because the test will only be used to determine the asset's expected market value, we henceforth assume without loss of generality that each test score is an unbiased estimate of the value.  That is, $S\equiv [\underline{\theta},\overline{\theta}]$ and for every market value $\theta$ in $\Theta$ and score $s$ in $S$, $E[\theta|T(\theta)=s]=s$. We denote the set of all tests by $\mathcal T$.
For each test $T$, let $G_T(s)$ denote the probability that if the intermediary evaluates the asset using test $T$, the resulting score is at most $s$. We refer to $G_T$ as the \emph{marginal score distribution}.

Turning to fees, we denote the pair of testing and disclosure fees by $\phi\equiv(\phit,\phid)$. The testing fee is the price the intermediary charges to evaluate and generate information about the asset and show it to the seller. The disclosure fee is the additional price to make that information hard or verifiable so that the agent can disclose it to the buyers.  

We denote the test-fee structure by $\tfstruct$. Each test-fee structure $\tfstruct$ induces a game $\game$ between the agent and the buyers with the following timeline:
\begin{enumerate}[noitemsep]
    \item Nature chooses the value $\theta$ of the asset according to the distribution $F$. 
    \item The agent, without observing $\theta$, decides whether to have the asset tested. If so, he pays the testing fee $\phit$ and observes a score $s$ drawn according to $T(\theta)\in\Delta S$.
    \item The agent decides whether to disclose the score to the market. If so, he pays the disclosure fee $\phid$ and the buyers observe the score. If he does not disclose, or if he does not take the test, the buyers observe a null message, $N$.
    \item The buyers bid simultaneously for the asset, and the asset is sold to the highest bidder, with ties broken uniformly. 
\end{enumerate}
The agent's payoff is the price at which he sells the asset minus any testing and disclosure fees that he pays. The payoff of the buyer that buys the asset at price $p$ is $\theta-p$; the other buyer's payoff is zero.
\footnote{We have framed our analysis in terms of this market game for concreteness. But the issues and our analysis apply to other settings in which an agent obtains evidence to persuade others. For example, a principal may decide how much to invest in a project based on her beliefs about the project's success, and the agent may acquire evidence to persuade the principal to invest more. Our model is isomorphic to such a setting if the principal's investment increases linearly in her posterior expectation and the agent's payoff increases linearly in the principal's investment.}

In the induced game $\game$, an agent's (behavioral) strategy is $(\sigma_{\testfee},\sigma_{\disclosurefee})$, where $\sigma_{\testfee}$ in $[0,1]$ is the probability with which he has the asset tested and $\sigma_{\disclosurefee}:S\rightarrow [0,1]$ is the score-contingent probability with which he discloses the score. Let $\hat{S}(T)\equiv\cup_{\theta\in\Theta} Im(T(\theta))$ denote the set of scores in the image of $T$. Buyer $i$'s (pure) strategy $\sigma_i:\hat{S}(T)\cup \{N\}\rightarrow \Re$ specifies buyer $i$'s bid following a disclosure of score $s$ or the null message $N$. A belief system $\rho_i:\hat{S}(T)\cup\{N\}\rightarrow \Delta \Theta$ specifies buyer $i$'s posterior belief about the asset's value following a disclosure of score $s$ or the null message $N$. We denote a strategy profile $\sigma=(\sigma_{\testfee},\sigma_{\disclosurefee},\sigma_1,\sigma_2)$ and belief system $\rho=(\rho_1,\rho_2)$ by $(\sigma,\rho)$. 

For a game $\game$, we denote the set of Perfect Bayesian Equilibria (henceforth PBE) by $\Sigma\tfstruct$.\footnote{A PBE corresponds to each player behaving in a way that is sequentially rational, beliefs $\mu_i$ being derived via Bayes Rule whenever possible, and off-path beliefs satisfying ``You can't signal what you don't know" \citep{fudenberg1991tirole}. The important implication for our setting is that both buyers have the same beliefs ($\rho_1(s)=\rho_2(s)$ for any score $s$) both on and off the equilibrium path.} Because the buyers have the same beliefs and compete in a first-price auction, each buyer's bid is equal to the asset's expected value given all the available information.

\subsection{Maximal Revenue Guarantees}\label{Section-ModelRPG}

We study the intermediary's \emph{maximal revenue guarantee}, namely the highest payoff that she can guarantee herself by choosing a test-fee structure, i.e., assuming that the equilibrium in the induced game is chosen adversarially to her interests. 
\begin{definition}\label{Definition-RobustPayoffGuarantee}
The intermediary's \textbf{maximal revenue guarantee} is
\begin{align}\label{Equation-RobustPayoffGuarantee}
    R_M\equiv\sup_{\tfstruct\in \mathcal T\times \Re^2}\, \inf_{(\sigma,\rho)\in \Sigma\tfstruct} \sigma_{\testfee}\left(\phit + \phid \int_S \sigma_{\disclosurefee}(s)dG_T\right).
\end{align}
\end{definition}
Recall that $\sigma_{\testfee}$ is the probability with which the agent pays for the intermediary to test the asset and $\sigma_{\disclosurefee}(s)$ is the probability with which the agent pays to disclose a score of $s$. The first term in the parenthesis in \eqref{Equation-RobustPayoffGuarantee} is the intermediary's revenue from the testing fee, and the second term is the revenue from the disclosure fee.

We compare this maximal revenue guarantee with the \emph{full informational surplus} that the intermediary could (conceivably) extract in equilibrium, which is $R_F\equiv \mu-\underline\theta$, where $\mu$ is the expected value of the asset. To see why, notice that because the market price of the asset is never lower than $\underline\theta$, the agent can guarantee himself that payoff by never taking the test. And since the total surplus in the economy is $\mu$, the intermediary cannot obtain more than $R_F$ in any equilibrium. 

Why do we study the maximal revenue guarantee? Although there are test-fee structures that have equilibria in which the intermediary obtains the full informational surplus, such test-fee structures have other equilibria in which the intermediary obtains $0$. We alluded to this fact in the introduction and in the example in \Cref{Section-Example}. In fact, for any test-fee structure in which the intermediary obtains close to the full informational surplus in some equilibrium, there is another equilibrium in which the intermediary obtains close to $0$. 
\begin{proposition}\label{prop-whyrobust}
There exists $\varepsilon^*>0$ such that for any $\varepsilon<\varepsilon^*$, if there exists an equilibrium of a test-fee structure $(T,\phi)$ in which the intermediary's revenue is at least $R_F-\varepsilon$, then there exists another equilibrium of that test-fee structure in which the intermediary's revenue is no more than $\delta(\varepsilon)$, where $\lim_{\varepsilon\rightarrow 0}\delta(\varepsilon)= 0$.

\end{proposition}

\cref{prop-whyrobust} formalizes the challenge that multiple equilibria present to the intermediary: choosing a test-fee structure because its most favorable equilibrium generates a high revenue leaves the intermediary vulnerable to an equilibrium that generates low or zero revenue. An immediate implication of \cref{prop-whyrobust} is that the maximal revenue guarantee, $R_M$, is bounded away from the full information surplus, $R_F$, by at least $\varepsilon^*$. Our main result, \cref{proposition-arbitrary-prior-exponential}, derives a tight bound on $R_M$ that, for every $\mu$, applies across distributions for which the asset's expected value is $\mu$. Because the proof of \cref{prop-whyrobust} uses techniques that we develop later in the paper, we do not develop the intuition here. The proof and a graphical intuition are in the Appendix.  

A separate rationale for focusing on the maximal revenue guarantee is that for any test-fee structure, the equilibrium that minimizes the intermediary's revenue also maximizes the agent's payoff. Thus, if the intermediary fears that the agent and the asset market will coordinate on the agent's preferred equilibrium, she would choose a test-fee structure that attains her maximal revenue guarantee.\footnote{Thus, in the game induced by a test-fee structure, the agent's preferred equilibrium is the uniquely Pareto efficient equilibrium from the perspective of players of that game.} 

Our analysis also shows how to maximize the intermediary's revenue across test-fee structures that admit a \emph{unique} equilibrium. The value of this related problem is at most $R_M$, but because the robustly optimal test-fee structure that we identify in \cref{proposition-arbitrary-prior-exponential} has a unique equilibrium, it is also a solution to this related problem.

\section{Preliminary Steps}\label{Section-MainResult}\label{Section-SimplifyingProblem}
To simplify the problem of solving for the maximal revenue guarantee, we take the following preliminary steps: 
\begin{enumerate}[noitemsep]
    \item We frame the analysis of tests purely in terms of marginal score distributions.
    \item We show that it suffices to consider only those test-fee structures in which the agent pays the testing fee with probability $1$ in every equilibrium.
    \item For every such test-fee structure, we characterize the equilibrium with the lowest revenue for the intermediary in terms of a score threshold for disclosure. 
    \item Because finding the maximal revenue guarantee involves optimizing with strict constraints, we formulate a relaxed problem with weak constraints that has the same solution.
\end{enumerate}

\subsection{Shifting from Tests to Marginal Score Distributions}

For the intermediary's revenue, all that matters about a test is its marginal score distribution: for all fees, if two tests have the same marginal score distribution, then they have the same equilibria.\footnote{Take two test-fee structures $\tfstruct$ and $(T',\phi)$ such that $T$ and $T'$ have the same marginal score distribution. Take an equilibrium $(\sigma,\rho)$ of $\game$. Suppose that $(\sigma,\rho)$ were played in $\mathcal G(T',\phi)$ and observe that given $\rho$ and that $G_T=G_T'$, the strategy profile $\sigma$ remains sequentially rational for the agent and the buyers. Moreover, given $\sigma$, $\rho$ continues to satisfy Bayes' Rule and the appropriate consistency condition. Therefore, $(\sigma,\rho)$ is an equilibrium of $\mathcal G(T',\phi)$.} Thus, henceforth, we refer to $\diststruct$ as a test-fee structure, where $G$ is a CDF on $[\underline{\theta},\overline{\theta}]$ that corresponds to the marginal score distribution of some test $T$ and $\phi$ is a pair of testing and disclosure fees. We denote by $\underline{s}_G$ the lowest score in the support of $G$.

Focusing on the set of marginal score distributions induced by all possible tests is useful because this set is easy to characterize. Recall that a distribution $G$ is a \emph{mean-preserving contraction} of $F$ if its support is in $[\underline\theta,\overline\theta]$ and $\int_{\underline{\theta}}^{s'} G(s)ds \leq \int_{\underline{\theta}}^{s'} F(s)ds$ for all $s' \in [\underline{\theta},\bar{\theta}]$, with equality at $s' = \bar{\theta}$. We denote by $\sosd(F)$ the set of distributions that are mean-preserving contractions of $F$. We then have the following classical result.\footnote{This result is in \cite{rothschild1970increasing} and \cite{blackwell1979theory}, and features in recent work on information design \citep[e.g.][]{gentzkow2016rothschild,roesler2017buyer}.
}

\setcounter{lemma}{-1}
\begin{lemma}\label{Proposition-MPC}
A marginal score distribution $G$ is induced by an (unbiased) test if and only if $G$ is in $\sosd(F)$.
\end{lemma}
We use this formulation to rewrite revenue guarantees: for a test-fee structure $(G,\phi)$, let $\hat\Sigma(G,\phi)$ be the equilibrium set of the induced game between the agent and the market. The \emph{revenue guarantee} of $(G,\phi)$ is the lowest revenue generated in any equilibrium of this test-fee structure: $R\diststruct\equiv\inf_{(\sigma,\rho)\in \hat\Sigma\diststruct} \sigma_{\testfee}\left(\phit + \phid \int_S \sigma_{\disclosurefee}(s)dG\right)$. The maximal revenue guarantee is thus $R_M=\sup_{\diststruct\in \sosd(F)\times \Re^2}R\diststruct$.

\subsection{Using Option-Value as a Carrot}
We show that in every test-fee structure, either the asset is tested with probability $1$ in every equilibrium or there exists an equilibrium in which the asset is tested with probability $0$ (and the intermediary's revenue is $0$). 
\begin{lemma}\label{lemma-worker-takes-test-wp1-in-all}
If a test-fee structure $\diststruct$ satisfies
	\begin{equation}
	\phit< \int_{\mu + \phid}^{\bar{\theta}} [s - (\mu+\phid)] dG,  	\label{eq:2}    \tag{P}
	\end{equation}
then the asset is tested with probability $1$ in every equilibrium; otherwise, there exists an equilibrium in which the asset is tested with probability 0.
\end{lemma}
The logic of \cref{lemma-worker-takes-test-wp1-in-all} is that \eqref{eq:2} is a \emph{participation constraint} that must hold for the asset to be tested with positive probability in every equilibrium. To see why, consider a test-fee structure $(G,\phi)$ and an equilibrium in which the asset is tested with probability $0$. Because the market expects non-disclosure with probability $1$, the price of the asset conditional on non-disclosure is $\mu$.

If the agent deviates and has the asset tested, then he optimally pays $\phid$ and discloses the score whenever it is higher than $\mu+\phid$. This is strictly profitable if
\begin{align*}
    \mu <-\phit + \underbrace{\int_{\underline\theta}^{\overline\theta}\max\{\mu,s-\phid\}}_{\text{Option Value}} dG,
\end{align*}
which, by re-arranging, yields \eqref{eq:2}. Thus, if \eqref{eq:2} holds, the asset is tested with positive probability in every equilibrium. The proof of \cref{lemma-worker-takes-test-wp1-in-all} shows if \eqref{eq:2} holds, the asset is in fact tested with probability $1$ in every equilibrium.

Using this result, we show that it suffices to restrict attention to non-negative fees: for any test-fee structure with a negative testing or disclosure fee, the intermediary can improve her revenue guarantee by using a test-fee structure that has non-negative fees. 
\begin{lemma}\label{lemma-nonnegative}
    For any test-fee structure $(G,\phi)$ with $\phid < 0$ or $\phit<0$, there exists a test-fee structure $(G',\phi')$ with non-negative fees such that $R(G,\phi)$ is strictly less than $R(G',\phi')$.
\end{lemma}

\subsection{Adversarial Disclosure Thresholds}

For any test-fee structure that satisfies \eqref{eq:2}, an adversarial equilibrium is one that minimizes the disclosure probability, but in which the asset is tested with probability $1$.
We show that such equilibria are threshold equilibria in which the agent does not disclose his score if he obtains a score at the threshold.

With \cref{lemma-worker-takes-test-wp1-in-all,lemma-nonnegative} in hand, we restrict attention to test-fee structures that satisfy \eqref{eq:2} and have non-negative disclosure fees. Given a test-fee structure $(G,\phi)$, consider thresholds $\thr$ that weakly exceed $\underline{s}_G$, where $\underline{s}_G$ is the lowest score in the support of $G$. We say that such a threshold $\thr$ is an \emph{equilibrium threshold} for a test-fee structure $(G,\phi)$ if
\begin{align}\label{Equation-EqmThresholdDefinition}
    \thr-\phid=E_{G}[s|s\leq \thr].
\end{align}
If $\thr$ satisfies \eqref{Equation-EqmThresholdDefinition}, then there is an equilibrium in which the asset is tested with probability $1$ and the agent discloses the score if and only if it strictly exceeds $\thr$. To see why, suppose that the agent behaves in this way and \eqref{Equation-EqmThresholdDefinition} holds. Then the LHS of \eqref{Equation-EqmThresholdDefinition} is the difference between the market price of the asset when a score of $\thr$ is disclosed and the disclosure fee, and the RHS is the market price for the asset when no score is disclosed. Thus, the agent is indifferent between disclosing and not disclosing a score of $\thr$. Because the market price following disclosure increases in the score but the market price following non-disclosure is constant in the score, the agent strictly prefers not to disclose scores lower than $\thr$ and to disclose scores higher than $\thr$. 

A test-fee structure may have multiple equilibrium thresholds. It may also have other equilibria, in which if the agent obtains a threshold score, he chooses to disclose his score with strictly positive probability rather than probability $0$. The mixed strategy equilibria for the fully revealing and the $3$-score tests in \Cref{Section-Example} are examples of such equilibria. We prove in \cref{lemma-highest-threshold} below that from the standpoint of adversarial equilibrium selection, it suffices to focus on equilibria in which the agent discloses his score only if \emph{strictly} exceeds the threshold. 

We proceed as follows. For each test-fee structure, we show that a highest equilibrium threshold exists and provide a characterization of the highest equilibrium threshold that we later use to find the robustly optimal test-fee structure. Finally, we show that this highest equilibrium threshold corresponds to an adversarial equilibrium in which the agent discloses his score if and only if it strictly exceeds the highest equilibrium threshold.
\begin{lemma}\label{lemma-highest-threshold}
If a test-fee structure $\diststruct$ satisfies \eqref{eq:2} and $\phid \geq 0$, then the following are true:  
\begin{enumerate}[nolistsep,label=\normalfont(\alph*)]
    \item \label{bullet-exist}A highest equilibrium threshold $\thr$ exists.
    \item \label{bullet-highest} A score threshold $\thr\geq\underline{s}_G$ is the highest equilibrium threshold if and only if 
    \begin{align}
\begin{aligned}
    \thr-\phid&=E_{G}[s|s\leq \thr],\\
    \thr'-\phid&> E_{G}[s|s\leq \thr']\qquad \forall \thr'>\thr.
\end{aligned}
\tag{HE}\label{eq:robustimplementation}	 
\end{align}
    \item \label{bullet-eqm}There exists an adversarial equilibrium in which the agent discloses score $s$ if and only if $s>\thr$ where $\thr$ is the highest equilibrium threshold.

\noindent 
\end{enumerate}
\end{lemma}

\Cref{figure-highest-threshold-characterization} illustrates our characterization of the highest equilibrium threshold.

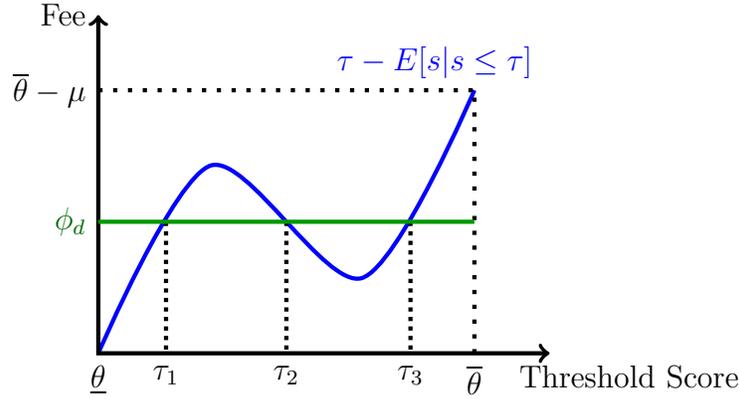
\begin{figure}[h]
	\centering
		\setlength{\tabcolsep}{-2pt}
	\begin{tikzpicture}[domain=0:3, scale=5, ultra thick,decoration={
		markings,
		mark=at position 0.5 with {\arrow{>}}}]    
	
	\draw [blue] plot [smooth] coordinates {(0,-.2) (0.3,.3) (.7,0) (1,.5)} node[above,xshift=-15]{$\thr - E[s | s \leq \thr]$};

	\draw[<->] (0,0.7) node[left]{Fee}-- (0,-.2) node[below]{$\lb$} -- (1.2,-0.2) node[below,xshift=30]{Threshold Score};

	\draw[dotted] (0.18,.15) -- (.18,-.2) node[below]{$\thr_1$};
	
	\draw[dotted] (0.83,.15) -- (.83,-.2) node[below]{$\thr_3$};

	\draw[dotted] (0.5,.15) -- (.5,-.2) node[below]{$\thr_2$};
	
		\draw[loosely dotted] (1,-0.2)node[below]{$\ub$} -- (1,.5);
	\draw[loosely dotted] (0,.5)node[left]{$\ub - \mu$} -- (1,0.5);
	
	\draw [green!60!black] (1,.15) -- (0,.15) node[left]{\textcolor{green!40!black}{$\phi_d$}};

	\end{tikzpicture}
	\caption{\footnotesize The worst equilibrium at a disclosure fee is that with the highest equilibrium threshold. For example, at a disclosure fee of $\phid$, the highest equilibrium threshold is $\thr_3$.}
	\label{figure-highest-threshold-characterization}
\end{figure}

\subsection{A Relaxation with an Identical Value}\label{section-problem-formulation}

We have seen that if a test-fee structure satisfies \eqref{eq:2}, then the agent takes the test with probability $1$ and the adversarial equilibrium is characterized by the highest equilibrium threshold. The revenue guarantee in the robustly optimal test-fee structure is therefore
\begin{align*}
    R_M=\sup_{(G,\phi,\thr)\in \sosd(F)\times \Re^3}\,\phit + \phid (1-G(\thr))
  &\qquad \text{s.t. } \eqref{eq:2}\text{ and }\eqref{eq:robustimplementation}.
\end{align*}
The intermediary's maximization problem has constraints with strict inequalities, but we show that the problem's value $R_M$ is unchanged if those inequalities are made weak. These are the \emph{weak participation constraint}
\begin{align}
\phit \leq \int_{\mu + \phid}^{\bar{\theta}} [s - (\mu+\phid)] dG\label{eq-phiT-relaxed-opt}\tag{w-P}
\end{align}
\noindent and the \emph{weak-highest equilibrium constraint} (which defines the weak-highest equilibrium threshold),
\begin{align}
\begin{aligned}
\thr-\phid&=E_{G}[s|s\leq \thr]\\
\thr' - \phid &\geq E_{G}[s|s\leq \thr']\qquad \forall \thr'>\thr.
\end{aligned}
\label{eq-relaxed-highest}\tag{w-HE}
\end{align}
We write the relaxed problem as
\begin{align*}
  \max_{(G,\phi,\thr)\in \sosd(F)\times \Re^3}  \phit + \phid(1-G(\thr))\qquad \text{s.t. } \eqref{eq-phiT-relaxed-opt}\text{ and }\eqref{eq-relaxed-highest}.
\end{align*}

In addition to having a solution and the same value as the original problem, the relaxed problem has several attractive features. First, for any test-fee structure that solves the relaxed problem, there is a ``nearby'' test-fee structure whose revenue guarantee is close to the maximal revenue guarantee $R_M$. Second, for any convergent sequence of test-fee structures that achieves the maximal revenue guarantee, the limiting test-fee structure is a solution to the relaxed problem. Thus, solutions to the relaxed problem identify necessary and sufficient features of test-fee structures whose revenue guarantees approximate $R_M$. We call these solutions \emph{robustly optimal test-fee structures}.

Let us formalize this discussion. A sequence of test-fee structures $\{(G^n,\phi^n)\}_{n=1,2,\ldots}$ converges to a test-fee structure $(G,\phi)$ if $G^n$ converges weakly to $G$ and $\phi^n$ converges to $\phi$. For a test-fee structure $(G,\phi)$ and threshold $\thr$, we denote the associated equilibrium revenue for the intermediary by $\hat R(G,\phi,\thr)$. We use this notation to link the relaxed and original problems. 

\begin{lemma}\label{lemma-relaxvalue}
\mbox{}
\begin{enumerate}
    \item[(a)] An optimal solution $(G,\phi,\thr)$ to the relaxed problem exists and $R_M = \hat{R}(G,\phi,\thr)$.
    \item[(b)]  Consider any optimal solution $(G,\phi,\thr)$ to the relaxed problem.  Then there exists a sequence of test-fee structures and thresholds $\{(G^n,\phi^n,\thr^n)\}_{n=1,2,\ldots}$ such that (i) for each $n$, $(G^n,\phi^n,\thr^n)$ satisfy \eqref{eq:2} and \eqref{eq:robustimplementation}, (ii) $(G^n,\phi^n)$ converges to $(G,\phi)$, and (iii) $R_M = \lim_{n\rightarrow\infty} \hat R(G^n,\phi^n,\thr^n)$.
    \item[(c)] Consider any sequence of test-fee structures and thresholds $\{(G^n,\phi^n,\thr^n)\}_{n=1,2,\ldots}$ such that (i) for each $n$, $(G^n,\phi^n,\thr^n)$ satisfy \eqref{eq:2} and \eqref{eq:robustimplementation}, (ii) $(G^n,\phi^n)$ converges to a test-fee structure $(G,\phi)$, and (iii) $R_M = \lim_{n\rightarrow\infty} \hat R(G^n,\phi^n,\thr^n)$.  Then there exists $\thr$ such that $(G,\phi,\thr)$ satisfy \eqref{eq-phiT-relaxed-opt} and \eqref{eq-relaxed-highest}, and $R_M = \hat R(G,\phi,\thr)$.
\end{enumerate}
\end{lemma}

\section{Robustly Optimal Test-fee Structures}\label{Section-RobustlyOptimalTestFeeStructures}

This section contains our main result. We show that robustly optimal test-fee structures use tests that have a ``step-exponential-step" form; i.e., the distributions over scores are exponential over an interval, have up to two mass points, one above and one below the interval, and have zero density everywhere else. We also  show that the optimal disclosure fees are strictly positive, and we derive a tight bound on the intermediary's maximal revenue guarantee.

\begin{figure}[h]
	\centering
	
	\begin{tikzpicture}[domain=0:3, scale=4, ultra thick,decoration={
		markings,
		mark=at position 0.5 with {\arrow{>}}}]    
	\draw[dotted] (0,1) node[left]{1} -- (1,1) -- (1,0) node[below]{$\ub$};

	\draw[<->] (0,1.2) node[left]{$G(s)$}-- (0,0) -- (1.2,0) node[below]{$s$};
	\draw (0,.05) -- (0,0) node[below]{$\lb$};

	\draw[line width=3,domain=.3:.7,smooth,variable=\x] (0,0) -- (0.1,0) node[below]{$\thr_0$} -- (.1,.2) -- (.3,.2) -- plot ({\x},{0.2*(e^((\x-0.3)/0.4))}) -- (.7,.5436) -- (0.85,.5436) -- (.85,1) -- (1,1);

	\draw[loosely dotted] (.3,.2) -- (.3,0) node[below]{$\thr_1$};
	
	\draw[loosely dotted] (.7,.5436) -- (.7,0) node[below]{$\thr_2$};
	
	\draw[loosely dotted] (.85,.5436) -- (.85,0) node[below]{$\thr_3$};
	
	\draw[dotted] (.3,.2) -- (0,.2) node[left]{$g$};
	
	\end{tikzpicture}
	\caption{\footnotesize A step-exponential-step distribution $G$.}
	\label{figure eq-G(x)-arbitary-prior} 
\end{figure}
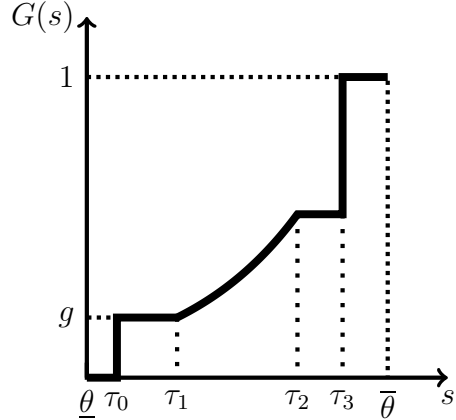
As illustrated in \autoref{figure eq-G(x)-arbitary-prior}, a test-fee structure $(G,\phi)$ is in the \emph{step-exponential-step} class if there exists $g \in [0,1]$ and thresholds $\thr_0<\thr_1<\thr_2\leq\thr_3$ such that 
\begin{equation}\label{eq-G(x)-arbitary-prior}
G(s)=\begin{cases}g &\text{ if }s \in  [\thr_0,\thr_1) \\
ge^{(s - \thr_1)/(\thr_1-\thr_0)} & \text{ if } s \in [\thr_1,\thr_2)\\
1 &\text{ if } s\geq \thr_3,
\end{cases}
\end{equation}
$G$ assigns probability $0$ to $[0,\thr_0)\cup (\thr_2,\thr_3)$, and the fees are 
\begin{align}
&\phid = \thr_1 - \thr_0,\label{eq-G(x)-arbitary-prior-feeD}\\
&\phit = (1-ge^{(\thr_2 - \thr_1)/\phid})(\thr_3 - (\mu + \phid)).\label{eq-G(x)-arbitary-prior-feeT}
\end{align}

\begin{proposition}\label{proposition-arbitrary-prior-exponential}
    For any distribution $F$ of the asset's value, the following hold:
    \begin{enumerate}[noitemsep,label=\normalfont(\alph*)]
        \item \label{Bullet-ExistenceSES}There exists a test-fee structure in the step-exponential-step class that is robustly optimal.
        
        \item \label{Bullet-StrictlyPositive} Every robustly optimal test-fee structure has a strictly positive disclosure fee.
        
        \item \label{Bullet-TestingFee} If $(G,\phi)$ is a robustly optimal test-fee structure, then the testing fee $\phit$ is strictly positive if and only if scores in $(\mu+\phid,\overline\theta]$ arise with positive probability.
        
        \item \label{Bullet-Surplus} The maximal revenue guarantee is at most $(\ub-\mu)(1-e^{{(\lb-\mu)}/{(\ub-\mu)}})$, and this bound is attained when the support of $F$ is binary, i.e., $\{\lb,\ub\}$.
    \end{enumerate}
\end{proposition}

\Cref{proposition-arbitrary-prior-exponential} partially characterizes the robustly optimal test-fee structures, and provides a tight upper bound on the maximal revenue guarantee. There is always a robust solution in the step-exponential-step class, disclosure fees are strictly positive in any robust solution, and testing fees are strictly positive only if the robustly optimal score distribution takes a particular form. We show in the proof of \Cref{proposition-arbitrary-prior-exponential}\ref{Bullet-ExistenceSES} that if the intermediary uses the robustly optimal test in the step-exponential-step class, charges the testing fee in \eqref{eq-G(x)-arbitary-prior-feeT} and a disclosure fee slightly below those in \eqref{eq-G(x)-arbitary-prior-feeD}, then that test-fee structure has a unique equilibrium. In that equilibrium, the asset is tested with probability $1$ and the agent discloses his score whenever it exceeds $\thr_1$. Our previous result, \cref{prop-whyrobust}, already established that the maximal revenue guarantee is bounded away from full surplus, but the bound provided was not tight; \ref{Bullet-Surplus} specifies a tight bound.

\Cref{proposition-arbitrary-prior-exponential} clarifies the interaction between the testing fee, the disclosure fee, and the test. \Cref{proposition-zero-disclosure-fee-then-full-revelation} in the Appendix shows that if the intermediary could not charge a disclosure fee, she would choose a fully revealing test and a testing fee that the agent pays with probability $1$. In the robustly optimal solution, the intermediary uses noisy tests because she can charge a disclosure fee. If she could not charge a testing fee, there would be a solution in the step-exponential-step class that would not have a step at the top (above the exponential part of the curve). Her ability to charge a testing fee in addition to a disclosure fee leads to a solution in the step-exponential-step class that may have a step at the top. 
\begin{figure}[t]
	\centering
	\setlength{\tabcolsep}{-2pt}
	\begin{tabular}{cc}
	
	\begin{tikzpicture}[domain=0:3, scale=5, ultra thick,decoration={
		markings,
		mark=at position 0.5 with {\arrow{>}}}]    
	
	\draw [blue] plot [smooth] coordinates {(1,0) (0.65,.4) (.35,0.3) (0,.7)};

    \draw[red] (1,0) -- (.95,.25) -- (.15,.25) -- (.15,.49) -- (0,.69);
    
    \draw[dotted] (.95,.25) -- (0,.25) node[left]{$\phid$};
    
    \draw[dotted] (.8,.25) -- (.8,0) node[below]{$q$};

	\draw[->] (.6,.38) -- (.6,.27);
	
	\draw[->] (.88,.2) -- (.94,.2);

	\draw[<->] (0,0.9) node[left]{Fee} -- (0,0) node[below,xshift=-3]{0} -- (1.2,0) node[below,xshift=2]{Prob. of};
	
    \draw (1.03,-.14) node[right]{Disclosure};

	\draw[loosely dotted] (1,0)node[below]{1} -- (1,.7);
	\draw[loosely dotted] (0,.7) node[left]{$\ub - \mu$} -- (1,0.7);

	\node at (0.5,-0.2) {(a)};

	\end{tikzpicture}
	&
	\begin{tikzpicture}[domain=0:3, scale=5, ultra thick,decoration={
		markings,
		mark=at position 0.5 with {\arrow{>}}}]

	\draw[<->] (0,0.9) node[left]{CDF} -- (0,0) node[below,yshift=1]{$\lb$} -- (1.2,0) node[below,xshift=-5]{Score};
	
	\draw[blue] (0,0) -- (1,.7);
	
	\draw[domain=5/12:.75,smooth,variable=\x, red] (0,0) -- plot ({\x-0.2},{(0.4*e^(2*(\x*1.7-1)))-.18}) -- (.76,.52) -- (1,.69);
	
	\draw[dotted] (.76,.52) -- (.76,0) node[below]{$\mu + \phid$};
	
	\draw[loosely dotted] (1,0)node[below,yshift=1]{$\ub$} -- (1,.7);
	\draw[loosely dotted] (0,.7) node[left]{1} -- (1,0.7);
	
	\draw[dotted] (.21,.14) -- (.21,0) node[below,yshift=-1]{$\thr$};

	\node at (0.5,-0.2) {(b)};

	\end{tikzpicture}
	\end{tabular}
	\caption{\footnotesize How an intermediary gains from ``flattening'' the robust demand curve. (a) depicts a demand curve in blue, and by flattening the demand curve the intermediary can induce a higher probability of disclosure at disclosure fee $\phid$. (b) shows that this can be done without changing the score distribution above $\mu+\phid$, which guarantees that \eqref{eq:2} continues to be satisfied without changing the testing fee.}
	\label{figure-arbitrary-priors-improvement}
\end{figure}
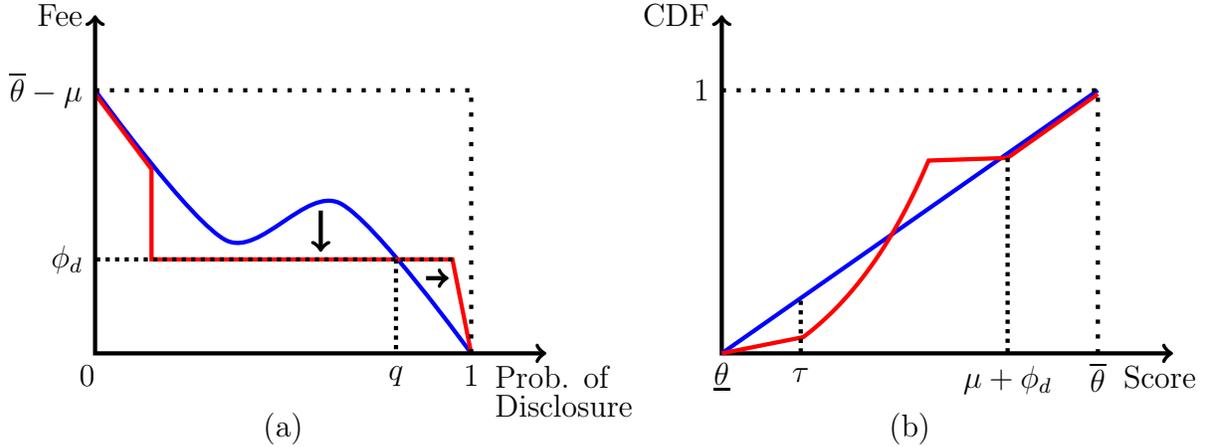

The intuition for \cref{proposition-arbitrary-prior-exponential}\ref{Bullet-ExistenceSES} builds on the robust demand curve approach described in \autoref{Section-Example}.  Consider a test-fee structure with disclosure fee $\phid$ in which the probability of disclosure in an adversarial equilibrium is $q$.  Suppose, as is shown in \autoref{figure-arbitrary-priors-improvement}(a), that the demand correspondence is not flat to the left of the point $(q,\phid)$.  Then we can modify the test-fee structure and improve the intermediary's robust revenue as follows: flatten the demand correspondence to the left of $(q,\phid)$ and push it to the right of $(q,\phid)$. Doing this increases the disclosure probability in the adversarial equilibrium for any disclosure fee slightly lower than $\phid$. Moreover, as seen in \autoref{figure-arbitrary-priors-improvement}(b), this modification can be done without changing the score distribution on scores above $\mu + \phid$, so the option value of a test, given by the RHS of \eqref{eq:2}, remains unchanged. 
Therefore, this modification does not generate a new adversarial equilibrium in which the agent doesn't pay the testing fee. 
Thus, the intermediary can increase the probability of disclosure by decreasing the disclosure fee arbitrarily slightly and without changing the testing fee, which improves her revenue guarantee.\footnote{A reader may wonder how this is possible while the total surplus is constant and \eqref{eq-phiT-relaxed-opt} (defined on p. \pageref{eq-phiT-relaxed-opt}) binds. This is because \eqref{eq-phiT-relaxed-opt} is a non-standard participation constraint. The modification to the test decreases the agent's payoff in the adversarial equilibrium without affecting his indifference between paying the testing fee and not paying it if the market expects him not to pay it, which is what \eqref{eq-phiT-relaxed-opt} represents.}

A similar logic shows that every robustly optimal test-fee structure has a score distribution that is exponential over a non-degenerate interval (\cref{lemma-exponential-middle} in the Appendix). Outside of that interval, there is some flexibility because scores below $\thr_1$ are those that the agent strictly prefers not to disclose, and scores above $\mu+\phid$, which exceeds $\thr_2$, are those for which \eqref{eq:robustimplementation} is slack. In the step-exponential-step distribution scores below $\thr_1$ are pooled as a single score and scores above $\mu+\phid$ are pooled as a single score. This creates a score distribution that is a mean-preserving contraction of the distribution $F$ of the asset value. But there may be other robustly optimal distributions that are the same on the interval $[\thr_1,\thr_2]$ and are also mean-preserving contractions of $F$. 

Turning to disclosure fees, we prove \Cref{proposition-arbitrary-prior-exponential}\ref{Bullet-StrictlyPositive} by showing that for the intermediary to charge a testing fee, the test must induce a high option value by having high scores in its support. But then the gain of adding a small disclosure fee outweighs the resulting reduction in the option value and the associated reduction in the testing fee.

\Cref{proposition-arbitrary-prior-exponential}\ref{Bullet-TestingFee} follows from the participation constraint \eqref{eq:2}. If scores in $(\mu+\phid,\overline\theta]$ arise with positive probability., then even if the market expects the asset to not be tested, the agent strictly prefers to have the asset tested at a testing fee of $0$ and a disclosure fee of $\phid$. Increasing the testing fee slightly while keeping the disclosure fee unchanged increases the intermediary's revenue in the adversarial equilibrium and does not generate a new equilibrium in which the asset is not tested. We use this logic in \Cref{proposition-log-concave-testing-fee-positive} below to provide conditions on primitives that guarantee a strictly positive testing fee. 

We prove \Cref{proposition-arbitrary-prior-exponential}\ref{Bullet-Surplus} by considering a relaxed problem in which the score distribution need not be a mean-preserving contraction of distribution $F$, while still having the same expectation, i.e., $E_G[s]=E_F[\theta]$. We solve this relaxed problem completely, and its value then provides an upper bound on the maximal revenue guarantee. This bound is tight when the support of $F$ is binary because the mean-preserving contraction condition is then equivalent to $E_G[s]=E_F[\theta]$. The following result describes the complete solution for this case. 
\begin{proposition}\label{prop-binaryrobustlyoptimal}
Suppose that the support of distribution $F$ is $\{\lb,\ub\}$. The unique robustly optimal test-fee structure involves: 
\begin{enumerate}[nolistsep,label=\normalfont(\alph*)]
    \item No testing fee but a strictly positive disclosure fee: $\phit^*=0$ and $\phid^*=\ub-\mu$.
    \item The following marginal score distribution that has an atom at $\lb$, a gap above it, and then an exponential form (with no atom at the top): 
   	\begin{equation}\label{eq-G(x)-binary}
G^*(s)=\begin{cases} e^{\frac{\lb-\mu}{\ub-\mu}}&\text{ if }s\in[\lb,\lb+\ub-\mu)\\
e^{\frac{s-\ub}{\ub-\mu}}&\text{ if }s\in[\lb+\ub-\mu,\ub].
\end{cases}
\end{equation}    
\end{enumerate}
\end{proposition}

We know from \Cref{proposition-arbitrary-prior-exponential} that the robustly optimal disclosure fee is always strictly positive; \Cref{prop-binaryrobustlyoptimal} shows that the robustly optimal testing fee may be $0$. We now provide a sufficient condition on the distribution $F$ for the optimal testing fee to be strictly positive. 

\begin{proposition}\label{proposition-log-concave-testing-fee-positive}
	If the distribution $F$ of the asset's value is log-concave, then any robustly optimal test-fee structure includes a strictly positive testing fee.
\end{proposition}

Several commonly studied distributions such as the uniform distribution over $[\lb,\ub]$ and the truncated Normal and Pareto distributions are log-concave \citep{bagnoli2005log}.  Log-concavity requires the tail of the distribution to be less heavy than the tail of the exponential distribution.

\section{Extensions}\label{Section-Extension}
This section shows that our results also hold when testing is costly for the intermediary and when the intermediary can offer the agent multiple pieces of evidence that he can disclose.

\subsection{Costly Tests}\label{section-extension-testingtechnology}

Our analysis assumed that testing is costless, but in reality testing is often costly. Suppose that the cost to the intermediary of running a test with marginal score distribution $G$ is $c(G)$. We assume that $c$ is lower semi-continuous in the weak* topology and monotone in the Blackwell order, that is, garbling a test weakly reduces its cost; in other words,  whenever $G'$ is a mean-preserving contraction of $G$, $c(G')\leq c(G)$.\footnote{Lower semi-continuity guarantees that solutions to the relevant maximization problems exist.} This condition is standard when information acquisition is costly, and corresponds to a less informative test being less costly to generate.\footnote{Recent examples of analyses that assume monotonicity in the Blackwell order are \cite*{de2017rationally} and \cite*{pomatto2019cost}.} A special case is when the intermediary starts with a finite set of initial tests and can garble those tests to obtain additional tests.  

The following is a corollary of our existing results.
\begin{proposition}\label{proposition-arbitrary-prior-exponential-adding-noise}
	If costs are lower semi-continuous and monotone in the Blackwell order, there exists a robustly optimal test-fee structure in the step-exponential-step class, and every robustly optimal test-fee structure uses a strictly positive disclosure fee.
\end{proposition}
\autoref{proposition-arbitrary-prior-exponential-adding-noise} is a corollary of \cref{lemma-exponential-middle} in the Appendix, which is one of the steps in the proof of \Cref{proposition-arbitrary-prior-exponential}. \cref{lemma-exponential-middle} shows that for any test-fee structure $(G,\phi)$, there exists a test-fee structure in the step-exponential-step class that uses a mean-preserving contraction of $G$ and has a weakly higher revenue guarantee. Because such a test-fee structure is a garbling of $G$, it has a weakly lower cost. Therefore, there must exist a robustly optimal test-fee structure in this class.\footnote{If costs are strictly monotone in the Blackwell order, then every robustly optimal test-fee structure is in the step-exponential-step class. } 
The argument used to prove \autoref{proposition-arbitrary-prior-exponential} proves that the optimal disclosure fee is strictly positive.

\subsection{Multiple Pieces of Evidence}\label{Section-Evidence}

Our analysis assumed that the intermediary provided the agent with a single piece of evidence that he could verifiably disclose. One could envision the intermediary providing the agent with multiple pieces of evidence and a choice of which to disclose. We show that this additional generality does not change the robustly optimal test-fee structures.

To see this, suppose that the intermediary designs, along with the test-fee structure, an arbitrary evidence structure, which specifies a message space $\mathcal{M}$ and the set of messages available for each score, described by $M: S\rightrightarrows \mathcal{M}$. We call this an \emph{evidence-test-fee structure}. Our baseline model corresponds to the special case in which $M(s)=\{s\}$ for every score $s$. The agent first decides whether to have the asset tested and pay the testing fee $\phit$. If he pays the testing fee, then he observes the test score $s$ and then chooses whether to disclose each message $m$ in $M(s)$ to the buyers. To disclose any of these messages, the agent pays the disclosure fee $\phid$ to the intermediary. If the agent does not have the asset tested or does not disclose any message, then the buyers observe the null message $N$. The following result shows that using evidence-test-fees structures does not improve the maximal revenue guarantee.

\begin{proposition}\label{prop-evidence}
For every adversarial equilibrium in the game induced by an evidence-test-fee structure, there exists a test-fee structure that has an adversarial equilibrium with the same revenue for the intermediary. 
\end{proposition}

If the intermediary chose the equilibrium in addition to the test-fee structure, \autoref{prop-evidence} would follow from the logic of the revelation principle. But adversarial equilibrium selection introduces a new consideration: because in the game induced by an evidence-test-fee structure the agent has more actions than in the game induced by a test-fee structure, it also has more deviations, and these deviations may exclude strategy profiles that would otherwise be equilibria. Thus, it is conceivable that under adversarial equilibrium selection some evidence-test-fee structure generates more revenue for the intermediary than any test-fee structure. We show this is not the case by finding a way to prune the evidence-test-fee structure to obtain a test-fee structure in a way that does not introduce any equilibria with lower revenue to the resulting game.

\section{Conclusion}\label{Section-Conclusion}

When assets are traded, it is common for sellers to disclose to buyers information obtained from a third party about the value of the assets. One rationale for the existence of such information intermediaries is that their presence generates economic value by, for example, mitigating moral hazard or facilitating assortative matching. A less obvious rationale for the presence of such intermediaries, even if they provide no economic value, is that once sellers of assets have the option to obtain hard information from an intermediary, potential buyers may have an unfavorable view of assets whose sellers do not disclose favorable information.

Our paper investigates the scope of this second rationale. In a setting in which information has no social value, we study how an intermediary designs and prices evidence for a disclosure game between an asset owner and the asset market. We show that even if the equilibria of the disclosure game are chosen to minimize the intermediary's revenue, she can still guarantee herself a high revenue across equilibria. We study how she accomplishes this.

First, she uses option value as a carrot. Because the agent prefers not to have the asset tested and the market to anticipate this, the intermediary chooses a test and fees so that the agent cannot credibly refrain from having the asset tested. We show that this corresponds to the intermediary creating option value by including high score realizations in the test and charging sufficiently low testing and disclosure fees. The market then correctly expect the agent to have the asset tested, and treats non-disclosure with prejudice. In this way, the intermediary exploits the agent's commitment problem to guarantee herself revenue. 

Second, the intermediary uses positive disclosure fees and noisy tests. Disclosure fees and noisy tests are closely related; were the intermediary constrained to charging only testing fees, she could optimally use tests that are fully revealing. To develop an intuition for the optimal combination of test and fees we propose a demand curve approach in which every test corresponds to a robust demand curve, and the intermediary can be thought of as choosing an optimal price on an optimal robust demand curve. The optimal disclosure fee maximizes the area of a rectangle under the robust demand curve, so the optimal robust demand curve has a rectangular component, which leads to an exponential distribution of scores. Finally, while we have focused on two-part tariffs, which are often used in practice, it would be interesting to study a broader range of pricing structures, including disclosure fees that depend on the realized test score and fees paid by prospective buyers of the asset.

\addcontentsline{toc}{section}{References}
\bibliographystyle{jpe}
\bibliography{certificates}

\addtocontents{toc}{\protect\setcounter{tocdepth}{1}} 
\appendix

\section{Appendix}\label{Section-Appendix}

\subsection{Proofs for \Cref{Section-Model}}\label{proof-Section-Model}
\subsubsection{Proof of \autoref{prop-whyrobust} on p. \pageref{prop-whyrobust}}

\begin{proof}
    
    Consider a test-fee structure $(G,\phi)$.  If \eqref{eq-phiT-relaxed-opt} is violated, then \autoref{lemma-worker-takes-test-wp1-in-all} shows that there is an equilibrium with zero revenue, and the proposition follows.
    Suppose that \eqref{eq-phiT-relaxed-opt} holds. Using integration by parts, we can rewrite \eqref{eq-phiT-relaxed-opt} as $\phit \leq  \int_{\mu + \phid}^{\ub} [1-G(s)]ds$. This expression implies that 

    any testing fee that satisfies \eqref{eq-phiT-relaxed-opt} is at most the area above the score 
    distribution $G$ from $\mu + \phid$ to $\ub$, shaded dark in 
    \autoref{figure-if-FSE-then-also-epsilon}.  The revenue from disclosure is at most $\phid \Pr[s 
    \geq \lb + \phid]$, shaded light in \autoref{figure-if-FSE-then-also-epsilon}.  This is because 
    in any equilibrium, a score strictly less than $\lb + \phid$ strictly prefers to conceal.  So the total revenue is at most the shaded area above $G$.  Since $G$ is a mean-preserving contraction 
    of the prior distribution, 
    \begin{align*}
        \mu = \int_{\lb}^{\ub} sdG(s) = \int_{\lb}^{\ub} [1-G(s)]ds + \ub
    \end{align*}
    \noindent and therefore the area above $G$ is equal to $R_F = \mu - \ub$.  
    \begin{figure}[h]
	\centering

	\begin{tikzpicture}[domain=0:3, scale=5, ultra thick,decoration={
		markings,
		mark=at position 0.5 with {\arrow{>}}}]    

    \filldraw[black!10] (.4,.28) -- (.4,.7) -- (0,.7) -- (0,.28) --cycle;
    
    \filldraw[black!40] (.7,.49) -- (1,.7) -- (.7,.7) -- cycle;

	\draw[<->] (0,0.9) node[left]{CDF} -- (0,0)-- (1.2,0) node[below,xshift=-5]{Score};
	
	\draw[blue] (0,0) node[below]{$\color{black} \lb$} -- (1,.7) node[above]{$G$};
	
	\draw[dotted] (.7,.49) -- (.7,0) node[below]{$\mu + \phid$};
	
	\draw[loosely dotted] (1,0)node[below]{$\ub$} -- (1,.7);
	\draw[loosely dotted] (0,.7) node[left]{$1$} -- (1,0.7);
	
	\draw[dotted] (.4,.28) -- (.4,0) node[below]{$\lb + \phid$};
	
	\end{tikzpicture}
	\caption{\footnotesize The revenue from testing fee fee is shaded dark.  The revenue from disclosure fee is at most the area shaded light.}
	\label{figure-if-FSE-then-also-epsilon}
\end{figure}
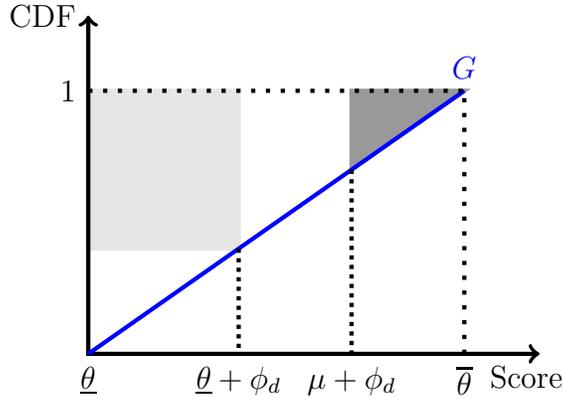
    
    Now suppose that there exists an equilibrium in which revenue is $R_F - \epsilon$.  Since revenue is at most the shaded area above $G$ and the total area above $G$ is $R_F$, then the un-shaded area above $G$ must be at most $\epsilon$.  In particular, the area above $G$ from $\lb + \phid$ to $\mu + \phid$ is at most $\epsilon$.  Since $G$ is monotone,
    \begin{align*}
        \phit \leq \int_{\mu + \phid}^{\ub} [1-G(s)]ds &\leq (\ub - (\mu + \phid))(1 - G(\mu + \phid)) \\ &\leq \left(\frac{\ub - (\mu + \phid)}{\mu - \lb}\right) \int_{\lb + \phid}^{\mu + \phid} [1-G(s)]ds \leq \left(\frac{\ub - \mu}{\mu - \lb}\right)\epsilon
    \end{align*}
    \noindent Therefore, as $\epsilon$ goes to zero, the revenue from testing goes to zero as well.  Thus to complete the proof we only need to show that the revenue from disclosure also goes to zero.
    
    For $\epsilon$ small enough so that $\mu + \phid \geq \lb + \phid+ \sqrt{\epsilon}$, we have
    \begin{align*}
        \epsilon \geq \int_{\lb + \phid}^{\mu + \phid} [1-G(s)]ds \geq \int_{\lb + \phid}^{\lb + \phid+ \sqrt{\epsilon}} [1-G(s)]ds \geq \sqrt{\epsilon}(1-G(\lb + \phid + \sqrt{\epsilon})),
    \end{align*}
    where the third inequality follows since $G$ is monotone.
    That is, the probability that the score is more than $\thr := \lb + \phid + \sqrt{\epsilon}$ is at most $\sqrt{\epsilon}$.  Thus if there exists an equilibrium threshold above $\thr$, the disclosure probability in that equilibrium is at most $\sqrt{\epsilon}$.  To show that there exists an equilibrium threshold above $\thr$, we apply \autoref{lemma-highest-threshold} by showing that $E[s | s \leq \thr] > \thr - \phid$.

    The expectation of $G$ can be written as
\begin{align*}
\mu&=G(\thr)E[s|s\leq \thr]+(1-G(\thr))E[s|s>\thr] \leq G(\thr)E[s|s\leq \thr]+(1-G(\thr))\ub.
\end{align*}
Rearranging terms yields
$$\ub-E[s|s\leq\thr]\leq \frac{\ub-\mu}{G(\thr)} \leq \frac{\ub-\mu}{1-\sqrt{\varepsilon}}.$$
Therefore since $\thr=\lb + \phid + \sqrt{\epsilon}$ and $\lb < \mu$, for small enough $\epsilon$ we have\footnote{Precisely, for $\varepsilon\leq \varepsilon^*\equiv (\frac{\mu-\lb}{1+\ub})^2$.}
$$\thr - \phid = \lb+\sqrt{\varepsilon} < \frac{\mu-\ub\sqrt{\varepsilon}}{1-\sqrt{\varepsilon}} = \ub-\frac{\ub-\mu}{1-\sqrt{\varepsilon}} \leq E[s|s\leq\thr],$$
and therefore by \autoref{lemma-highest-threshold} there must exist an equilibrium threshold higher than $\thr$.

To complete the proof, recall that $\phid<\ub-\mu$, so the revenue from disclosure fee is no more than $\sqrt{\varepsilon}(\ub-\mu).$ Thus the total revenue is bounded by
$\delta(\varepsilon)\equiv(\frac{\ub - \mu}{\mu - \lb})\epsilon+\sqrt{\varepsilon}(\ub-\mu).$
\end{proof}

\subsection{Proofs for \Cref{Section-SimplifyingProblem}}
\subsubsection{Proof of \autoref{lemma-worker-takes-test-wp1-in-all} on p. \pageref{lemma-worker-takes-test-wp1-in-all}.}

\begin{proof}
	Consider a test-fee structure $\diststruct$ that satisfies \eqref{eq:2}.  Assume towards a contradiction that there exists an equilibrium in which the agent has the asset tested with probability strictly less than $1$. In this equilibrium, let $p_N$ be the price that the agent obtains when he does not disclose his score. In any equilibrium, a agent discloses the score if $s-\phid>p_N$ and does not disclose his score if $s-\phid<p_N$. Therefore, the expected value conditional on non-disclosure is no higher than the prior expected value: $p_N\leq\mu$.
	
	Because the agent has the asset tested with probability strictly less than $1$, his equilibrium payoff equals $p_N$. Consider the following deviation: having the asset tested with probability 1, and disclosing if and only if the score $s$ satisfies $s> p_N+\phid$. The payoff from this deviation is $ -\phit+\int_{p_N+\phid}^{\overline\theta} (s-\phid)dG+\int_{\underline\theta}^{p_N+\phid}p_N dG $ where the first term is the testing fee, the second is the payoff from disclosing a high score, and the third is the payoff from concealing a low score. Taking the difference between this deviation payoff and the equilibrium payoff of $p_N$ yields
	\begin{align*}
	-\phit+\int_{p_N+\phid}^{\overline\theta} (s-\phid)dG+\int_{\underline\theta}^{p_N+\phid}p_N dG-p_N&=-\phit+\int_{p_N+\phid}^{\overline\theta} (s-p_N-\phid)dG\\
	&\geq -\phit+\int_{\mu+\phid}^{\overline\theta} (s-\mu-\phid)dG\\
	&> 0.
	\end{align*}
\noindent where the equality follows from algebra, the first inequality follows from $p_N\leq \mu$, and the second inequality follows from \eqref{eq:2}. Therefore, the deviation is strictly profitable, and this strategy profile is not an equilibrium.

Now consider a test-fee structure $\diststruct$ that violates \eqref{eq:2}, and where
\begin{align}\label{Equation-ParticipationViolated}
    \phit\geq \int_{\mu+\phid}^{\overline\theta}[s-(\mu+\phid)]dG.
\end{align}

\noindent	We show that now there exists an equilibrium where the agent has the asset tested with probability $0$. Consider a strategy profile in which the agent's strategy involves having the asset tested with probability 0, and therefore, he has no opportunity to disclose on the path of play. Off-path, the agent discloses when obtaining a score $s\geq \mu+\phid$.  

In this strategy profile, the market price from non-disclosure is $\mu$. We claim that this strategy profile is an equilibrium. First, the agent's behavior is sequentially rational off-path. Second, by deviating to take the test, the agent receives an expected payoff lower than $\mu$:
	\begin{align*}
	 \int_\lb^{\mu+\phid}\mu dG+\int_{\mu+\phid}^\ub [s-\phid]dG-\phit&=\mu+\int_{\mu+\phid}^\ub[s-(\mu+\phid)]dG-\phit\leq \mu
	\end{align*}
where the equality is algebra, and the inequality follows from \eqref{Equation-ParticipationViolated}. So the agent has no profitable deviation, and the strategy profile is an equilibrium.

\end{proof}

\subsubsection{Proof of \autoref{lemma-nonnegative} on p. \pageref{lemma-nonnegative}}

\begin{proof}

Consider a test-fee structure $(G,\phi)$ where at least one of the fees is negative.

If $\phit-\int_{\mu+\phid}^\ub[s-(\mu+\phid)]dG<0$, from \autoref{lemma-worker-takes-test-wp1-in-all}, the asset is tested with probability 1 in all equilibria. Suppose $\phid<0$, then in the unique equilibrium all scores disclose, so $R(G,\phi)=\phit+\phid$. Let $\phid'=0$ and $\phit'=\phit+(1-G(\mu+\phid))\phid$. Then $ \phit'-\int_{\mu+\phid'}^\ub[s-(\mu+\phid')]dG\leq \phit-\int_{\mu+\phid}^\ub[s-(\mu+\phid)]dG<0$, so we can find $\phit''>\phit'$ and $\phid''=\phid'=0$ such that $\phit''-\int_{\mu+\phid''}^\ub[s-(\mu+\phid'')]dG<0$. Then  \autoref{lemma-worker-takes-test-wp1-in-all} implies that under $(G,\phi'')$ the asset is tested with probability 1 and $R(G,\phi'')=\phit''>\phit'=\phit+(1-G(\mu+\phid))\phid\geq \phit+\phid=R(G,\phi)$. Now if $\phit''\geq 0$, we have constructed a test-fee structure with non-negative fees that generates a strictly higher revenue. If $\phit''<0$, since $\phid''=0$, it implies $R(G,\phi'')\leq 0$. But then we can easily construct a test-fee strcuture $(G''',\phi''')$ that generates a strictly positive revenue: let $G'''=F$ and $\phid'''=0$. By assumption $F$ is not degenerated (i.e., $\lb<\ub$), so 
$\int_{\mu}^\ub (s-\mu)  dF>0.$
Let $\phit'''\in(0,\int_{\mu}^\ub (s-\mu)  dF)$, and from \autoref{lemma-worker-takes-test-wp1-in-all} in all equilibria the asset is tested with probability 1 under $(G''',\phi''')$, and $R(G''',\phi''')=\phit'''>0$. Now suppose $\phit<0$ and $\phid\geq 0$. Let $\phit''''=0$ and $\phid''''=\phid$, then $(G,\phi'''')$ satisfies \eqref{eq:2} unless  $\int_{\mu+\phid}^\ub[s-(\mu+\phid)]dG=0$. So if $\int_{\mu+\phid}^\ub[s-(\mu+\phid)]dG>0$,  a higher testing fee is charged while the revenue from disclosure fee is unchanged, because the  disclosure decision doesn't depend on the testing fee (this can also be seen from $\phit$ does not appear in \eqref{eq:robustimplementation}). If $\int_{\mu+\phid}^\ub[s-(\mu+\phid)]dG=0$, then $G(\mu+\phid)=1$, which implies $\thr=\mu+\phid$ is the threshold satisfying \eqref{eq:robustimplementation}, and $R(G,\phi)=\phit+0<0$. But then the $(G''',\phi''')$ constructed above gives a strictly higher revenue with non-negative fees.

If $\phit-\int_{\mu+\phid}^\ub[s-(\mu+\phid)]dG\geq 0$, from \autoref{lemma-worker-takes-test-wp1-in-all}, $R(G,\phid)\leq 0$. Then the $(G''',\phi''')$ constructed above gives a strictly higher revenue with non-negative fees.

\end{proof}

\subsubsection{Proof of \autoref{lemma-highest-threshold} on p. \pageref{lemma-highest-threshold}}

The following lemma is used in our proof.

\begin{lemma}\label{lemma-inter}
	Suppose $f$ is an increasing function defined on $[a,b]\subset \mathbb{R}$, and $g$ is a continuous function defined on $[a,b]$. If $f(a)>g(a)$ and $f(b)<g(b)$, there exists $x^*\in(a,b)$ such that $f(x^*)=g(x^*)$. Moreover, $\bar{x}=\max\{x|f(x)=g(x)\}$ exists and $\bar{x}=\sup\{x\in[a,b]|f(x)\geq g(x)\}$.
\end{lemma}
\begin{proof}
Define $S=\{x\in[a,b]|f(x)\geq g(x)\}$. Becase $S$ is non-empty and bounded above by $b$, there exists a unique supremum $\bar{x}=\sup S$. Since $f$ is increasing, $f(a)>g(a)$, and $g$ is continuous, it follows that for sufficiently small $\varepsilon>0$, $f(a+\varepsilon)>f(a)>g(a+\varepsilon)$, which implies $\bar{x}> a$.
Consider an increasing sequence $x_n\in  S$ such that $x_n\rightarrow \bar{x}$. Since $f$ is an increasing function,  $f(\bar{x})\geq \limsup_{n\rightarrow \infty} f(x_n)$. Recall that $f(x_n)\geq g(x_n)$, so taking limsup for both sides yields
\begin{align*}
    f(\bar{x})\geq \limsup_{n\rightarrow \infty}f(x_{n})\geq \limsup_{n\rightarrow \infty}g(x_{n})=g(\bar{x})
\end{align*}
where the equality holds because $g$ is continuous. So $f(\bar{x})\geq g(\bar{x})$, which also implies $\bar{x}\neq b$.

	Now we prove that $f(\bar{x})= g(\bar{x})$. Suppose towards a contradiction that $f(\bar{x})> g(\bar{x})$. Because $g$ is continuous, there exists $\varepsilon>0$ such that $f(\bar{x})> g(\bar{x}+\varepsilon)$. Since $f$ is increasing, it then follows that
	$$f(\bar{x}+\varepsilon)\geq f(\bar{x})> g(\bar{x}+\varepsilon)$$
	which means $\bar{x}+\varepsilon\in S$, contradicting $\bar{x}=\sup S$. Therefore, it follows that $\bar{x}=\max\{x|f(x)=g(x)\}$.    
\end{proof}

\begin{proof}[Proof of \autoref{lemma-highest-threshold}\ref{bullet-exist}]
 	A threshold $\thr$ is an equilibrium if and only if
	$$E_G[s|s\leq \thr]=\thr-\phid.$$
	
	Define $a(\thr)=E_G[s|s\leq \thr]$ and $b(\thr)=\thr-\phid$. Since $a(\thr)$ is bounded by $\mu$ and $b(\thr)\rightarrow\infty$ as $\thr\rightarrow\infty$, there exists $\bar{\thr}$ large enough such that $a(\thr)<b(\thr)$ for all $\thr\geq \bar{\thr}$. So if an equilibrium exists, the threshold $\thr$ must be in $[\underline{s}_G,\bar{\thr}]$.

	We have $a(\underline{s}_G)=\underline{s}_G\geq b(\underline{s}_G)$ and $a(\bar{\thr})\leq b(\bar{\thr})$, $a$ is increasing and $b$ is continuous. From \autoref{lemma-inter}, there exists $\thr\in[\underline{s}_G,\bar{\thr}]$ such that $a(\thr)=b(\thr)$, and $\max\{\thr|a(\thr)=b(\thr)\}$ exists. So the set of equilibrium thresholds is non-empty and a highest equilibrium threshold exists.
\end{proof}

\begin{proof}[Proof of \autoref{lemma-highest-threshold}\ref{bullet-highest}]
    The ``if" part follows directly from definition, so we prove the ``only if" part. Suppose that $\thr$ is the highest threshold of $(G,\phi)$, so $\thr-\phid=E[s|s\leq \thr]$ and for all $\thr'>\thr$, $\thr'-\phid\neq E[s|s\leq \thr']$.  We show that indeed we must have $\thr'-\phid > E[s|s\leq \thr']$.
	
	Suppose for contradiction that there exists $\thr'>\thr$ such that $\thr'-\phid<E[s|s\leq \thr']$. Since $a(\thr')>b(\thr')$, $a(\bar{\thr})\leq b(\bar{\thr})$, $a$ is increasing and $b$ is continuous, from \autoref{lemma-inter}, there exists $\thr^*\in(\thr',\bar{\thr}]\subset (\thr,\bar{\thr}]$ such that $a(\thr^*)=b(\thr^*)$.  It implies threshold $\thr^* > \thr$ is an equilibrium, which leads to a contradiction. Therefore, for all $\thr'>\thr$, we must have $\thr'-\phid> E[s|s\leq \thr'].$
\end{proof}

\begin{proof}[Proof of \autoref{lemma-highest-threshold}\ref{bullet-eqm}]
    We show that the agent discloses and only discloses scores greater than the highest threshold $\thr$ is an adversarial equilibrium. 
	Suppose there exists another equilibrium that gives the intermediary a lower revenue. In this equilibrium, let $\tilde{\thr}=p_N+\phid$, where $p_N$ is the price of no disclosure. Since all scores strictly greater than $\tilde{\thr}$ must disclose and all scores strictly lower than $\tilde{\thr}$ must conceal in equilibrium, the equilibrium price satisfies 
	$$E[s< \tilde{\thr}]\leq p_N \leq E[s|s\leq \tilde{\thr}].$$
	
	The intermediary's revenue from this equilibrium is $\phit+\phid[1-G(\tilde{\thr})+\lambda(G(\tilde{\thr})-\sup_{s<\thr}G(s))]$,
	for some $\lambda\in[0,1]$, which denotes the probability that the agent discloses when $s=\thr$. Since this equilibrium gives the intermediary a lower revenue, 
	$$1-G(\tilde{\thr})+\lambda(G(\tilde{\thr})-\sup_{s<\thr}G(s))< 1-G(\thr),$$
	which implies $\tilde{\thr}>\thr$. But then from part (b) the characterization of $\thr$ and the fact that $\tilde{\thr}>\thr$,
	$$\tilde{\thr}-\phid>E[s|s\leq\tilde{\thr}]\geq p_N $$
	which contradicts $\tilde{\thr}=p_N+\phid$.
\end{proof}

\subsubsection{Proof of \autoref{lemma-relaxvalue} on p. \pageref{lemma-relaxvalue}}

We prove \autoref{lemma-relaxvalue} by first proving  part (c), then (a), and then (b).

We use the following lemma to prove \autoref{lemma-relaxvalue}(c).

\begin{lemma}\label{lemma-usc}
    	If $(G^n,\thr^n)\rightarrow (G,\thr)$, and $\lim_{n\rightarrow \infty}G^n(\thr^n)$ exists, then $\lim_{n\rightarrow \infty}G^n(\thr^n)\leq G(\thr)$.
\end{lemma}
\begin{proof}
    It suffices to show that for any $c>G(\thr)$, there exists $N$ such that for $n\geq N$, $G^n(\thr^n)< c$. Because $G$ is right-continuous, it follows that for any $c>G(\thr)$, there exists $\bar\varepsilon>0$ such that 
	\begin{equation}\label{eq-rightcontinuous-1}
	c>G(\thr+2\bar\varepsilon)+\bar\varepsilon
	\end{equation}
Notice that $\thr^n\rightarrow \thr$, so there exists $N_1$ such that for $n>N_1$, $\thr^n\leq \thr+\bar\varepsilon$, which implies
	\begin{equation}\label{eq-rightcontinuous-2}
	G^n(\thr^n)\leq G^n(\thr+\bar\varepsilon)
	\end{equation}
	Because weak convergence is metrized by the Levy metric,
 	$$L(G,F)=\inf\{\varepsilon>0|F(x-\varepsilon)-\varepsilon\leq G(x)\leq F(x+\varepsilon)+\varepsilon\text{ for all }x\}.$$
	Thus, if $G^n$ converges weakly to $G$, there exists $N_2$ such that for all $n>N_2$, $L(G,G^n)<\bar\varepsilon$, and hence
	\begin{equation} \label{eq-rightcontinuous-3}
	G(\thr+2\bar\varepsilon)+\bar\varepsilon\geq G^n(\thr+\bar\varepsilon)
	\end{equation}
Combining \eqref{eq-rightcontinuous-1}, \eqref{eq-rightcontinuous-2} and \eqref{eq-rightcontinuous-3}, it follows that for every $n \geq \max\{N_1,N_2\}$,
\begin{align*}
c>G(\thr+2\bar\varepsilon)+\bar\varepsilon\geq G^n(\thr+\bar\varepsilon)\geq G^n(\thr^n).
\end{align*}
\end{proof}

\begin{proof}[Proof of \autoref{lemma-relaxvalue}(c)] 
Let us begin by verifying \eqref{eq-phiT-relaxed-opt}. We have

    $$\phit=\lim_{n\rightarrow \infty}\phit^n\leq \lim_{n\rightarrow \infty}\int_{\mu+\phid^n}^\ub[s-(\mu+\phid^n)]dG^n(s)=\int_{\mu+\phid}^\ub[s-(\mu+\phid)]dG(s)$$
    where the equalities follow from taking limits and the inequality follows since $(G^n,\phi^n,\thr^n)$ satisfies \eqref{eq:2}.  So $(G,\phi)$ satisfies \eqref{eq-phiT-relaxed-opt}.

    Now we turn to \eqref{eq:robustimplementation}. Since $\thr_{n}-\phi_d^n=E[s|s\leq\thr^n]\leq \mu$ and $\phid^n\rightarrow \phid$, we have $\thr^n$ being bounded above. On the other hand $\thr^n\geq \underline{s}_G\geq \lb$ so $\thr^n$ is bounded below. From Bolzano-Weierstrass theorem, there exists a subsequence $n_k$ and $\thr$ such that $\lim_{k\rightarrow \infty}\thr_{n_k}=\thr$. 	
	
	For any $\thr'>\thr$, we show that $\phid G(\thr')\leq \int_\lb^
	{\thr'} G(s)ds$. If $\phid G(\thr')> \int_\lb^{\thr'} G(s)ds$, there exists small enough $\varepsilon$ such that $\phid G(\thr'+\varepsilon)> \int_\lb^{\thr'+\varepsilon} G(s)ds$ and $G$ is continuous at $\thr'+\varepsilon$. So $\lim_{n\rightarrow \infty}\phid^nG^n(\thr'+\varepsilon)=\phid G(\thr'+\varepsilon)$. But then we have a contradiction:
		$$\phid G(\thr'+\varepsilon)>\int_\lb^{\thr'+\varepsilon} G(s)ds=\lim_{k\rightarrow\infty}\int_\lb^{\thr'+\varepsilon} G_{n_k}(s)ds\geq \lim_{k\rightarrow \infty } \phid^{n_k} G_{k_n} (\thr'+\varepsilon)=\phid G(\thr'+\varepsilon) $$
	where the second inequality holds because $\thr'>\thr$ implies for large enough $k$, $\thr'+\varepsilon>\thr'>\thr_{n_k}$ and thus $ \phid^{n_k} G_{n_k}(\thr'+\varepsilon)\leq \int_\lb^
	{\thr'+\varepsilon} G_{n_k}(s)ds$. 
	
	Next we show that $\phid G(\thr)=\int_\lb^\thr G(s)ds$. Since $G$ is right continuous and for all $\thr'>\thr$, $\phid G(\thr')\leq \int_\lb^
	{\thr'} G(s)ds$, we must have $\phid G(\thr)\leq \int_\lb^\thr G(s)ds$. Since  $\phid^{n_k} G_{n_k}(\thr_{n_k})=\int_\lb^{\thr_{n_k}} G_{n_k}(s)ds$ for all $k$, $\lim_{k\rightarrow\infty}\phid^{n_k}G_{n_k}(\thr_{n_k})$ exists and equals $\lim_{k\rightarrow\infty}\int_\lb^{\thr_{n_k}} G_{n_k}(s)ds=\int_\lb^{\thr} G(s)ds$. Therefore
	$$\phid G(\thr)\geq \lim_{k\rightarrow\infty}\phid^{n_k}G_{n_k}(\thr_{n_k})=\int_\lb^{\thr} G(s)ds$$
	where the inequality holds from \autoref{lemma-usc}. Notice that the argument above also implies $\phid G(\thr)= \lim_{k\rightarrow\infty}\phid^{n_k}G_{n_k}(\thr_{n_k})$.
	
	Now we have shown $(G,(\phid,\phit))$ and $\thr$ satisfies \eqref{eq-phiT-relaxed-opt} and \eqref{eq-relaxed-highest}, and next we consider the revenues comparison:
	$$V=\lim_{n\rightarrow \infty} \phit^n+\phid^n(1-G^n(\thr^n))= \phit+\lim_{k\rightarrow \infty} \phid^{n_k}(1-G_{n_k}(\thr_{n_k}))=\phit+\phid(1- G(\thr)).$$
	
	So $(G,(\phid,\phit))$ and $\thr$ indeed generate a revenue as the limit of the revenues generated by $(G^n,(\phid^n,\phit^n))$ and $\thr^n$.
\end{proof}

We use the following lemma to prove \autoref{lemma-relaxvalue}(a) and \autoref{lemma-relaxvalue}(b).

\begin{lemma}\label{lemma-relaxvalue-4b-old}
For every test-fee structure and threshold $(G,\phi,\thr)$ that satisfy \eqref{eq-phiT-relaxed-opt} and \eqref{eq-relaxed-highest}, there exists a sequence of test-fee structures and thresholds $\{(G^n,\phi^n,\thr^n)\}_{n=1,2,\ldots}$ such that (i) for each $n$, $(G^n,\phi^n,\thr^n)$ satisfy \eqref{eq:2} and \eqref{eq:robustimplementation}, (ii) $(G^n,\phi^n)$ converges to $(G,\phi)$, and (iii) $\hat R(G,\phi,\thr)\leq \lim_{n\rightarrow\infty} \hat R(G^n,\phi^n,\thr^n)$.
\end{lemma}
\begin{proof} 
    Consider a test-fee structure $(\phit,\phid,G)$ and $\thr$ satisfying \eqref{eq-phiT-relaxed-opt} and \eqref{eq-relaxed-highest}. Therefore,
\begin{align*}
\phit&\leq \int_{\max\{\mu+\phid,\ub\}}^\ub (s-\mu-\phid)dG(s)\text{ and}\\
\thr-\phid&=E[s|s\leq \thr]\\
\thr'-\phid&\geq E[s|s\leq \thr']\quad\text{for all }\thr'>\thr.
\end{align*}

Let $\phid^n=\phid-\frac{1}{n}$ and $\phit^n=\phit-\frac{1}{n}$, then 
$$\phit^n<\phit\leq \int_{\mu+\phid}^\ub (s-\mu-\phid)dG(s)\leq  \int_{\mu+\phid^n}^\ub (s-\mu-\phid^n)dG(s),$$
so \eqref{eq:2} is satisfied.

From \autoref{lemma-highest-threshold}, we know that under $(\phit^n,\phid^n,G)$, the highest equilibrium threshold $\hat{\thr}_n$ exists. Moreover, since
 \begin{align*}
\thr-\phid^n&>\thr-\phid=E[s|s\leq \thr]\\
\thr'-\phid^n&>\thr'-\phid\geq E[s|s\leq \thr']\quad\text{for all }\thr'>\thr.
\end{align*}
It implies the highest threshold $\hat{\thr}_n<\thr$. So the revenue under $(\phit^n,\phid^n,G)$ is 
$$\phit^n+\phid^n(1-G(\hat{\thr}_n))\geq \phit^n+\phid^n(1-G(\thr))\geq\phit+\phid(1-G(\thr))-\frac{2}{n}$$
where the last term goes to $\hat{R}(G,\phi,\thr)$ as $n\rightarrow \infty$.
\end{proof}

\begin{proof}[Proof of \autoref{lemma-relaxvalue}(a)]

Consider any sequence $(G^n,(\phit^n,\phid^n))$ and $\thr^n$ satisfying \eqref{eq:2} and \eqref{eq:robustimplementation} and generating value $V_n\rightarrow R_M$.

The disclosure fee being non-negative ($\phid^{n}\geq 0$) implies $\phit^{n}\leq \int_{\mu}^\ub (s-\mu)dG^n(s)\leq \ub-\mu$. For any $\phid^{n}>\ub-\mu$, \eqref{eq:2} and \eqref{eq:robustimplementation} implies $\phit^n<0$ and $G(\thr^n)=1$, and thus $V_n<0$. So without loss of generality, we can consider $\phid^{n}\leq \ub-\mu$. Moreover, for any $\phit^n<-1$, by replacing it with $\phit^n=-1$, \eqref{eq:2} and \eqref{eq:robustimplementation} are still satisfied, and $V^n$ increases. So without loss of generality, we can restrict attention to $\phit^n\geq -1$.

Now we have a sequence $(G^n,(\phit^n,\phid^n))$ such that $\phit^{n_j}\in[-1,\ub-\mu]$, $\phid^{n_j}\in[0,\ub-\mu]$, $G^n\in\Gamma(F)$. Proposition $1$ of \cite{kleiner2020extreme} proves that $\Gamma(F)$ is compact, which implies that there exists a converging subsequence $(G^{n_k},(\phit^{n_k},\phid^{n_k}))\rightarrow (G,(\phit,\phid))$. From \autoref{lemma-relaxvalue}(c), we can find $\thr$ such that $(G,(\phit,\phid))$ and $\thr$ satisfies the constraints \eqref{eq-phiT-relaxed-opt} and \eqref{eq-relaxed-highest} in the relaxed problem, and generates a revenue $R_M$. Moreover, from \autoref{lemma-relaxvalue-4b-old}, any test-fee structure and threshold satisfying \eqref{eq-phiT-relaxed-opt} and \eqref{eq-relaxed-highest} must at most generate a revenue of $R_M$, so $(G,(\phit,\phid))$ and $\thr$ is an optimal solution to the relaxed problem.  
\end{proof}

\begin{proof}[Proof of \autoref{lemma-relaxvalue}(b)] 
By \autoref{lemma-relaxvalue-4b-old}, there exists a sequence of test-fee structures and thresholds $\{(G^n,\phi^n,\thr^n)\}_{n=1,2,\ldots}$ such that (i) for each $n$, $(G^n,\phi^n,\thr^n)$ satisfy \eqref{eq:2} and \eqref{eq:robustimplementation}, (ii) $(G^n,\phi^n)$ converges to $(G,\phi)$, and (iii) $\lim_{n\rightarrow\infty} \hat R(G^n,\phi^n,\thr^n) \geq \hat R(G,\phi,\thr) = R_M$.  By definition of $R_M$, we also have $\lim_{n\rightarrow\infty} \hat R(G^n,\phi^n,\thr^n) \leq R_M$.  Therefore, $\lim_{n\rightarrow\infty} \hat R(G^n,\phi^n,\thr^n) = R_M$.
\end{proof}

\subsection{Proofs for \autoref{Section-RobustlyOptimalTestFeeStructures}}

This section contains the proofs of \Cref{proposition-arbitrary-prior-exponential,prop-binaryrobustlyoptimal,proposition-log-concave-testing-fee-positive}.  The three statements of \Cref{proposition-arbitrary-prior-exponential} build on each other in reverse order.  We therefore prove \Cref{proposition-arbitrary-prior-exponential} in reverse order, beginning with \ref{Bullet-Surplus} and ending with \ref{Bullet-ExistenceSES}. \Cref{prop-binaryrobustlyoptimal} is a corollary of one of the steps (\Cref{lem:the-relaxed-problem}) used in \Cref{proposition-arbitrary-prior-exponential}\ref{Bullet-Surplus}. 

\subsubsection{Proof of \autoref{proposition-arbitrary-prior-exponential}\ref{Bullet-Surplus} on p. \pageref{proposition-arbitrary-prior-exponential}}

We prove \autoref{proposition-arbitrary-prior-exponential}\ref{Bullet-Surplus} by solving a relaxed problem where the mean-preserving contraction constraints are relaxed to requiring the score distribution to have same expectation as the prior mean:
\begin{align}\label{eq-mpc-relaxed}
    \max_{(G,\phi,\thr)  \in\Delta[\lb,\ub]\times \Re^3} \hat R(G,\phi,\thr) \qquad\text{s.t. }\eqref{eq-phiT-relaxed-opt}, \eqref{eq-relaxed-highest},\text{ and }E_G[s]=E_F[s]  \tag{RE}
\end{align}
Recall that $\hat R(G,\phi,\thr)\equiv\phit + \phid(1-G(\thr))$ is the revenue corresponding to a test-fee structure $(G,\phi)$ and a threshold $\thr$. This is a relaxed problem because $G\in\Gamma(F)$ implies $G\in\Delta[\lb,\ub]$ and $E_G[s]=E_F[s]$.

We use the following lemma to solve the relaxed problem. The lemma shows that in any optimal testing fee structure, the only possible score above $\mu+\phid$ is $\ub$.  In other words, the score distribution $G$ is flat from $\mu+\phid$ to $\ub$, with possibly a discrete jump at $\ub$ so that $G(\ub) = 1$.

\begin{lemma}\label{lemma-flat-above-mu-plus-c}
	If a test-fee structure $(G,\phi)$ with a weak-highest equilibrium threshold $\tau$ is an optimal solution to the relaxed problem \eqref{eq-mpc-relaxed}, then $G(s) = G(\mu+\phid)$ for all $s \in [\mu+\phid,\ub)$.
\end{lemma}
\begin{proof}
	Using integration by parts, the testing fee of an optimal test-fee structure is

$$	\phit = \int_{\mu + \phid}^\ub [s - (\mu+\phid)] dG(s)  = - \int_{\mu + \phid}^\ub G(s)ds + (\ub-(\mu+\phid)) $$
\noindent Thus, revenue is
$$	\hat R(G,\phi,\thr) = -\phid G(\thr) - \int_{\mu + \phid}^\ub G(s)ds + (\ub-\mu). $$
\noindent Since $\thr$ is a weak-highest equilibrium threshold, $\phid G(\thr) = \int_\lb^{\thr}G(s)ds$.  Substituting this equality into revenue, we have
\begin{align}
    	\hat R(G,\phi,\thr) &= - \int_\lb^{\thr} G(s)ds - \int_{\mu + \phid}^\ub G(s)ds + (\ub-\mu) = \int_{\thr}^{\mu + \phid} G(s)ds,\label{lemma-flat-above-mu-plus-c-eq1}
\end{align}

	\noindent where the second equality followed from the constraint that the integral of $G$ is $1-\mu$.
	
	Now consider an optimal test-fee structure $(G,\phi)$ with a weak-highest equilibrium threshold $\thr$, and suppose for contradiction that $G(s) > G(\mu+\phid)$ for some  $s \in [\mu+\phid,\ub)$.  Construct a distribution $G'$ as follows.  Let $G'(s) = \alpha G(s)$ for some $\alpha$ and all $s \leq \mu+\phid$, and $G'(s) = G'(\mu+\phid)$ for all $s \in [\mu+\phid,\ub)$.  By the assumption that $G(s) > G(\mu+\phid)$ for some  $s \in [\mu+\phid,\ub)$, there exists $\alpha > 1$ such that the areas under $G$ and $G'$ are equal.  Define $\phid' = \phid$, and $\phit'$ so that the upper bound on the testing fee \eqref{eq-phiT-relaxed-opt} holds with equality for distribution $G'$.
	
	We now show that $\thr$ is a weak-highest equilibrium threshold of the test-fee structure $(G',\phi')$, and gives higher revenue than $(G,\phi)$.  Indeed, for any $\thr'$ such that $\thr \leq \thr' \leq \mu + \phid$, since $G'$ is equal to $G$ multiplied by $\alpha$, we have
	\begin{align*}
	\frac{\int_\lb^{\thr'} G(s)ds}{G(\thr')} = \frac{\int_\lb^{\thr} G(s)ds}{G(\thr)}.
	\end{align*}
	\noindent Thus, from integration by parts, $\thr' - E_{G'}[s | s \leq \thr'] = \thr' - E_{G}[s | s \leq \thr'] \geq \phid$, with equality at $\thr' = \thr$.  In addition, for all $\thr' \geq \mu+\phid$, since $G'$ is flat we have 
	\begin{align*}
	\frac{\int_\lb^{\thr'} G'(s)ds}{G'(\thr')} \geq \frac{\int_\lb^{\mu+\phid} G(s)ds}{G(\mu+\phid)} \geq \phid.
	\end{align*}
	\noindent It remains to show that the revenue of $(G',\phi')$ is higher than that of $(G,\phi)$.  This fact follows directly from \eqref{lemma-flat-above-mu-plus-c-eq1}, since $G' > G$ below $\mu + \phid$.\footnote{If $\tau=\mu+\phid$, $\hat{R}(G,\phi,\thr)=0$. As we have shown in the proof of \autoref{lemma-nonnegative}, the optimal revenue is strictly positive. So in an optimal test-fee structure $(G,\phi)$ with threshold $\tau$, $\tau<\mu+\phid$.}
\end{proof}

Equipped with \autoref{lemma-flat-above-mu-plus-c}, we solve the relaxed problem \eqref{eq-mpc-relaxed}.

\begin{lemma}\label{lem:the-relaxed-problem}

The value of \eqref{eq-mpc-relaxed} is $(\ub-\mu)(1-e^{\frac{\lb-\mu}{\ub-\mu}})$, and it has a unique solution $\phit^*=0$, $\phid^*=\ub-\mu$, and
   	\begin{equation*}
G^*(s)=\begin{cases} e^{\frac{\lb-\mu}{\ub-\mu}}&\text{ if }s\in[\lb,\lb+\ub-\mu)\\
e^{\frac{s-\ub}{\ub-\mu}}&\text{ if }s\in[\lb+\ub-\mu,\ub].
\end{cases}
\end{equation*}    
\end{lemma}

\begin{proof}
We first argue that the test-fee structure $(G^*,\phi^*)$ achieves the revenue bound $(\ub-\mu)(1-e^{\frac{\lb-\mu}{\ub-\mu}})$.  We will later show that $(\ub-\mu)(1-e^{\frac{\lb-\mu}{\ub-\mu}})$ is an upper bound on revenue and therefore $(G^*,\phi^*)$ is optimal.  	We start by verifying that $(G^*,\phi^*)$ is a feasible solution.

  \paragraph{Weak participation constraint is satisfied} Since $\mu + \phid^* =\ub$, the weak participation constraint \eqref{eq-phiT-relaxed-opt} holds with equality, i.e.,
	\begin{align*}
	\int_{\mu + \phid^*}^\ub [s - (\mu+\phid^*)] dG^*(s) = 0 = \phit^*.
	\end{align*}
	
	\paragraph{The constraint $E_G[s]=E_F[\theta]$ is satisfied} 		Directly from the definition of $G^*$, for any $\thr' \geq \lb+\ub-\mu$ we have
	\begin{align*}
	\int_\lb^{\thr'} G^*(s)ds &= (\ub-\mu)e^{\frac{\lb-\mu}{\ub-\mu}} + (\ub-\mu)e^{\frac{s-\ub}{\ub-\mu}}\bigg|_{\lb+\ub-\mu}^{\thr'} \nonumber = (\ub-\mu)e^{\frac{\thr'-\ub}{\ub-\mu}}.
	\intertext{Thus, in particular, $\int_\lb^{\ub} G^*(s) = \ub-\mu$ and hence, by integration by parts,} 
	E_{G^*}[s] &= \int_\lb^\ub sdG^*(s)= \ub-\int_\lb^\ub G^*(s)ds =\mu.
	\end{align*}
	\noindent Therefore, $E_G[s]=E_F[\theta]$.
	
	\paragraph{Weak-highest equilibrium constraint is satisfied} We show that $\thr = \lb+\ub-\mu$ is a weak-highest equilibrium threshold, i.e., it satisfies \eqref{eq-relaxed-highest}. For any $\thr'\geq \lb+\ub-\mu$,
	\begin{align*}
	E[s|s\leq \thr']&=\frac{\int_\lb^{\thr'} sdG^*(s)}{G^*(\thr')}=\frac{\thr'G(\thr')-\int_\lb^{\thr'} G^*(s)ds}{G^*(\thr')}= \thr'-\frac{(\ub-\mu)e^{\frac{\thr'-\ub}{\ub-\mu}}}{e^{\frac{\thr'-\ub}{\ub-\mu}}}=\thr'-\phid^*.
\end{align*}

\noindent Therefore, $\thr = \lb+\ub-\mu$ satisfies \eqref{eq-relaxed-highest} and is a weak-highest equilibrium threshold.
	
	\noindent  The lemma below shows that $E[s] = \mu$ and $\thr = \lb+\ub-\mu$ is a weak-highest equilibrium threshold.  
	
	Thus, the revenue of the test fee structure $(G^*,\phi^*)$ with weak-highest equilibrium threshold $\thr = \lb+\ub-\mu$ is
	\begin{align*}
	\phit^* + \phid^*(1-G(\lb+\ub-\mu)) = 0 +(\ub-\mu)(1-e^{\frac{\lb-\mu}{\ub-\mu}}).
	\end{align*}
	
	\paragraph{Upper bound on revenue} It remains to show that the revenue of any test-fee structure is at most $(\ub-\mu)(1-e^{\frac{\lb-\mu}{\ub-\mu}})$.
	
	For any $\thr'$ such that $\thr \leq \thr' \leq \mu+\phid$, the total area under $G$ is
	\begin{align}
	\ub-\mu &= \int_\lb^{\thr'} G(s) ds + \int_{\thr'}^{\ub} G(s) ds \nonumber \\
	& \leq \phid e^{\frac{1}{\phid}(\thr'-\thr)}G(\thr) + \int_{\thr'}^{\ub} G(s) ds \nonumber \\
	& \leq \phid e^{\frac{1}{\phid}(\thr'-\thr)}G(\thr) + (\ub-\thr')G(\mu+\phid), \nonumber
	\intertext{where the first inequality followed from \autoref{lemma-boundspeed}, and the second inequality followed since $G(s) \leq G(\mu+\phid)$ for all $s \leq \mu+\phid$, and $G(s) = G(\mu+\phid)$ for all $s \geq \mu+\phid$ by \autoref{lemma-flat-above-mu-plus-c}. Let $\thr' = \phid+\frac{\mu-\ub(1-G(\mu+\phid))}{G(\mu+\phid)}$.   Notice that $\thr'\in[\thr,\mu+\phid]$ because from the definition of $\thr$, $\thr=\phi_d+E[s|s\leq \thr]\leq \phi_d+E[s|s\leq \mu+\phi_d]=\thr'$, and $\mu<\ub$ implies $\frac{\mu-\ub(1-G(\mu+\phid))}{G(\mu+\phid)}<\mu$.  Therefore, by the above discussion we have}
	G(\thr) &\geq \frac{1}{\phid}e^{-\frac{1}{\phid}(\thr'-\thr)}((\ub-\mu) - (\ub-\thr')G(\mu+\phid)) \nonumber \\
	&= \frac{1}{\phid}e^{-\frac{1}{\phid}(\thr'-\thr)}((\ub-\mu) - (\frac{\ub-\mu}{G(\mu+\phid)}-\phid)G(\mu+\phid)) \nonumber \\
	& =  e^{-\frac{1}{\phid}(\thr'-\thr)} G(\mu+\phid) \nonumber \\
	& \geq  e^{-\frac{1}{\phid}(\ub-\lb-\frac{\ub-\mu}{G(\mu+\phid)})} G(\mu+\phid),\label{theorem-profit-supremum-binary-eq1}
	\end{align}
	\noindent where the last inequality followed since $\lb+\phid \leq \thr$.

	We now use \eqref{theorem-profit-supremum-binary-eq1} and \autoref{lemma-flat-above-mu-plus-c} to bound revenue.  Since by \autoref{lemma-flat-above-mu-plus-c}, the distribution is flat above $\mu+\phid$, the testing fee is
	\begin{align*}
	\phit & =  \int_{\mu + \phid}^\ub [s - (\mu+\phid)] dG(s) = (1-G(\mu+\phid))(\ub-(\mu+\phid))
	\end{align*}
	\noindent Thus revenue is
	\begin{align*}
	\hat R(G,\phi,\thr)\leq \phid\left[1-e^{-\frac{1}{\phid}(\ub-\lb-\frac{\ub-\mu}{G(\mu+\phid)})} G(\mu+\phid)\right] + (1-G(\mu+\phid))(\ub-(\mu+\phid)).
	\end{align*}
	The above expression is non-decreasing in $\phid$.  Since $\phid \leq \ub-\mu$, an upper bound on revenue is obtained by substituting $\phid = \ub-\mu$ into the above expression, which yields
	\begin{align*}
	  (\ub-\mu)(1-e^{-\frac{\ub-\lb}{\ub-\mu}+\frac{1}{G(\mu+\phid)}} G(\mu+\phid)).  
	\end{align*}
	This expression is non-decreasing in $G(\mu+\phid)$.  Since $G(\mu+\phid) \leq 1$, an upper bound on revenue is obtained by substituting $G(\mu+\phid) =1$ into the above expression, which yields $(\ub-\mu)(1-e^{\frac{\lb-\mu}{\ub-\mu}})$, 
	completing the proof.
\end{proof}

\begin{proof}[Proof of \autoref{proposition-arbitrary-prior-exponential}\ref{Bullet-Surplus}]
    The solution to the relaxed problem in \autoref{lem:the-relaxed-problem} is $(\ub-\mu)(1-e^{\frac{\lb-\mu}{\ub-\mu}})$.  Therefore $(\ub-\mu)(1-e^{\frac{\lb-\mu}{\ub-\mu}})$ is an upper bound on the revenue of any test-fee structure. Moreover, if the support of the prior is binary (i.e. $\{\lb,\ub\}$), $G\in\Gamma(F)$ is equivalent to $G\in\Delta[\lb,\ub]$ and $E_G[s]=E_F[s]$, so the bound is attained in this case.
\end{proof}

\subsubsection{Proof of \autoref{proposition-arbitrary-prior-exponential}\ref{Bullet-TestingFee} on p. \pageref{proposition-arbitrary-prior-exponential}}

\begin{proof}[Proof of \autoref{proposition-arbitrary-prior-exponential}\ref{Bullet-TestingFee}]
  In a robustly optimal test-fee structure, the constraint  \eqref{eq-phiT-relaxed-opt} must bind, otherwise the intermediary can strictly increases the revenue by increasing $\phit$, so $\phit=\int_{\mu + \phid}^{\bar{\theta}} [s - (\mu+\phid)] dG$. Notice that $\int_{\mu + \phid}^{\bar{\theta}} [s - (\mu+\phid)] dG>0$ iff $G(\mu+\phid)<1$, which then implies the result.

\end{proof}

\subsubsection{Proof of \autoref{proposition-arbitrary-prior-exponential}\ref{Bullet-StrictlyPositive} on p. \pageref{proposition-arbitrary-prior-exponential}}

We prove \autoref{proposition-arbitrary-prior-exponential}\ref{Bullet-StrictlyPositive} in two steps.  We first identify test-fee structures that are robustly optimal among ones with zero disclosure fees.  We show that it is in fact robustly optimal to use a fully revealing test.  We will then show that such a test-fee structure can be improved upon using a small disclosure fee.

\begin{lemma}\label{proposition-zero-disclosure-fee-then-full-revelation}
	The test-fee structure $(F,\phi)$ with $\phit=\int_{\mu}^\ub[s-\mu]dF(s)$ and $\phid = 0$ is robustly optimal among all test-fee structures with zero disclosure fee.
\end{lemma}
\begin{proof}
	
	Consider a test-fee structure with zero disclosure fee. The revenue is
	\begin{align*}
	&\int_{\mu}^\ub (s-\mu)dG(s) = (\ub-\mu)-\int_\mu^\ub G(s)ds.
	\end{align*}
	\noindent Thus revenue is maximized by minimizing $\int_\mu^\ub G(s)ds$.
	
	The MPC constraints require that $\int_{\thr}^\ub G(s)ds\geq \int_{\thr}^\ub F(s)ds$ for all $\thr\in[\lb,\ub]$. Specifically at $\thr=\mu$, we have
	$$\int_\mu^\ub G(s)ds\geq \int_\mu^\ub F(s)ds.$$
	Thus the optimal revenue is at most
	$$(\ub-\mu)-\int_\mu^\ub F(s)ds$$
	which is obtained by setting $G=F$.
\end{proof}

Equipped with \autoref{proposition-zero-disclosure-fee-then-full-revelation}, we prove \autoref{proposition-arbitrary-prior-exponential}\ref{Bullet-StrictlyPositive}.

\begin{proof}[Proof of \autoref{proposition-arbitrary-prior-exponential}\ref{Bullet-StrictlyPositive}]
Consider the revenue from charging a positive disclosure fee $\phid'$, and reducing the testing fee to $\phit'$ so that the upper bound on testing fee \eqref{eq-phiT-relaxed-opt} binds, ensuring participation with probability one.  This change involves a loss and a gain in revenue.  The loss is from lower testing fee, which is the difference between $\ub-\mu - \int_\mu^\ub G(s)ds$ and $\ub-(\mu+\phid') - \int_{\mu+\phid'}^\ub G(s)ds$. By algebra, the loss is at most $\phid'(1-G(\mu))$.

The gain from charging a positive disclosure fee is $\phid'(1-G(\thr))$, where $\thr$ is a weak-highest equilibrium threshold.  Notice that since $\lb$ is in the support of the distribution, we have $G(\thr')>0$ for all $\thr' > \lb$.  Thus, by integration by parts we have 
	\begin{align*}
	\thr' - E[s | s \leq \thr'] = \frac{\int_\lb^{\thr'} G(s)ds}{G(\thr')} > 0.
	\end{align*}
	\noindent for all $\thr' >\lb$.  Therefore, for any $\thr' > \lb$, there exists a small enough disclosure fee such that any weak-highest equilibrium threshold is below $\thr'$.  In particular, for a small enough disclosure fee, any weak-highest equilibrium threshold $\thr$ is below $\mu$ and satisfies $G(\thr) \leq G(\mu)$.  As a result, for small enough disclosure fee, the gain is at least the loss.	Further, if $G(s) > G(\mu)$ for some $s \in [\mu,\ub)$, then there exists $\phid$ such that the loss is strictly less than $\phid'(1-G(\mu))$.  	This is the case if $G$ is non-binary.  As a result, for a non-binary $G$, there exists $\phi'$, $\phid' > 0$, such that the revenue of $(G,\phi')$ is strictly larger than the revenue of $(G,\phi)$.
	
	Suppose first that the prior $F$ is non-binary.  By \autoref{proposition-zero-disclosure-fee-then-full-revelation}, the test-fee structure $(F,\phi)$ where $\phid = 0$ is optimal among all test-fee structures with zero disclosure fee.  The argument above shows that there exists $\phi'$ with positive disclosure fee and such that the revenue of $(F,\phi')$ is strictly higher than that of $(F,\phi)$, and hence of any test-fee structure with zero disclosure fee.
	
	Now suppose that the prior $F$ is binary.     With binary support, the mean-preserving contraction constraints become $E_G[s]=E_F[\theta]$.  Therefore, by \autoref{lem:the-relaxed-problem}, the optimal revenue is $(\ub-\mu)(1-e^{\frac{\lb-\mu}{\ub-\mu}})$.  By \autoref{proposition-zero-disclosure-fee-then-full-revelation}, the optimal revenue among test-fee structures with zero disclosure fee is obtained by full revelation.  The revenue is the testing fee
	\begin{align*}
	\ub-\mu - \int_{\mu}^{\ub} F(s)ds = (\ub-\mu) - \frac{(\ub-\mu)^2}{\ub-\lb} =\frac{ (\ub-\mu)(\mu-\lb)}{\ub-\lb},
	\end{align*}
	which is strictly less than the optimal revenue $(\ub-\mu)(1-e^{\frac{\lb-\mu}{\ub-\mu}})$ for all $\mu\in(\lb,\ub)$.  Since $\ub$ is in the support of the distribution, it ensures that $\mu\in(\lb,\ub)$.
\end{proof}

\subsubsection{Proof of \autoref{proposition-arbitrary-prior-exponential}\ref{Bullet-ExistenceSES} on p. \pageref{proposition-arbitrary-prior-exponential}}

We use two lemmas in the proof of \autoref{proposition-arbitrary-prior-exponential}\ref{Bullet-ExistenceSES}.  The first lemma bounds the rate at which the integral of a score distribution can grow given the \eqref{eq-relaxed-highest} constraint.

\begin{lemma}\label{lemma-boundspeed}
	Suppose that $\phid > 0$.  Let $\thr$ be a weak-highest equilibrium threshold. Then for any $\thr_a$ and $\thr_b$ where $\thr \leq \thr_a \leq \thr_b$,
	$$\int_{\lb}^{\thr_b}G(s)ds\leq e^{\frac{1}{\phid}(\thr_b-\thr_a)}\int_\lb^{\thr_a }G(s)ds,$$ with equality if and only if \eqref{eq-relaxed-highest} holds with equality for all thresholds in $[\thr_a,\thr_b]$.
\end{lemma}
\begin{proof}
 	Using integration by parts, the weak-highest equilibrium threshold constraint \eqref{eq-relaxed-highest} can be written as
	\begin{align}
	&\phid \leq \thr' - E[s | s \leq \thr'] = \frac{\int_\lb^{\thr'} G(s)ds}{G(\thr')},\nonumber
	\intertext{for all $\thr' \geq \thr$, with equality at $\thr' = \thr$.  The right hand side is the inverse of the derivative of $\ln(\int_\lb^{\thr'} G(s)ds)$ with respect to $\thr'$.  Thus,}
	&\frac{d}{d\thr'}(\ln(\int_\lb^{\thr'} G(s)ds))\leq \frac{1}{\phid}. \label{eq-derivative-of-log-upper-bound}
	\intertext{Integrating from $\thr_a$ to $\thr_b$,}
	&\ln(\int_\lb^{\thr_b} G(s)ds) - \ln(\int_\lb^{\thr_a} G(s)ds) \leq \frac{1}{\phid}(\thr_b-\thr_a).\nonumber
	\intertext{Raising both sides to the power of $e$,}
	&\frac{\int_\lb^{\thr_b} G(s)ds}{\int_\lb^{\thr_a} G(s)ds} \leq e^{\frac{1}{\phid}(\thr_b-\thr_a)},\nonumber
	\end{align}
	\noindent with equality if and only if \eqref{eq-relaxed-highest} holds with equality for all thresholds in $[\thr_a,\thr_b]$.
\end{proof}

The following lemma is the main step in the proof of \autoref{proposition-arbitrary-prior-exponential}\ref{Bullet-ExistenceSES}.  It shows that for any test-fee structure, there exists a test-fee structure that is in the step-exponential-step class and has a weakly higher revenue guarantee.  The lemma further establishes that \emph{any} robustly optimal test must be exponential over an interval.

\begin{lemma}\label{lemma-exponential-middle}
        For any test-fee structure $(G,\phi)$ with $\phid > 0$ and a corresponding weak-highest threshold $\tau_1$, there exists a mean-preserving contraction $G'$ of $G$, a fee structure $\phi'$, and a threshold $\tau_1'$ such that the test-fee structure $(G',\phi')$ is in the step-exponential-step class, $\tau_1'$ is a weak-highest threshold and the robust revenue of $(G',\phi',\tau_1')$ is at least the robust revenue of $(G,\phi,\tau_1)$.      Further,  if $(G,\phi)$ is robustly optimal and has weak-highest equilibrium threshold $\thr_1$, then $\phid > 0$ and there exists a threshold $\thr_2\in[\thr_1,\mu+\phid]$, such that $G$ is exponential from $\thr_1$ to $\thr_2$, i.e., 
    \begin{align*}
        G(s) = G(\thr_1)e^{\frac{1}{\phid}(s - \thr_1)}, \forall s \in [\thr_1,\thr_2],
        \intertext{and is flat from $\thr_2$ to $\mu + \phid$, i.e.,}
        G(s) = G(\mu + \phid), \forall s \in [\thr_2,\mu+\phid].
    \end{align*}
\end{lemma}

\begin{proof}
Given a test-fee structure $(G,\phi)$ with a weak-highest threshold $\tau_1$, clearly by definition of $\tau_1$, we have $\tau_1\leq \mu+\phid$. Moreover, if $\tau_1=\mu+\phid$, the intermediary receives zero revenue, which is never optimal. So we consider $\tau_1<\mu+\phid$.

We first consider the case $G(\mu+\phid)=G(\tau_1)$, where $\tau_1$ is a weak-highest equilibrium threshold under $(G,\phi)$. Define $\tau_a=E_G[s|s\leq \tau_1]$, and $\tau_b=E_G[s|s> \mu+\phid]$. From the definition of $\tau_1$, we have $\tau_a=E_G[s|s\leq \tau_1]=\tau_1-\phid$. The mean constraint requires $G(\tau_1)\tau_a+(1-G(\tau_1))\tau_b=\mu$, which implies $G(\tau_1)=\frac{\tau_b-\mu}{\tau_b-\tau_a}$.

Given the test $G$, and to satisfy \eqref{eq-phiT-relaxed-opt}, the highest testing fee the intermediary can charge is $\phit=\int_{\mu+\phid}^1[s-(\mu+\phid)]dG(s)$. The disclosure fee being charged under equilibrium threshold $\tau_1$ is $\phid(1-G(\tau_1))$. So the intermediary's robust revenue bound under test $G$ is
\begin{align*}
    \int_{\mu+\phid}^\ub[s-(\mu+\phid)]dG(s)+\phid(1-G(\tau_1))&=\ub-\mu-\phid-\int_{\mu+\phid}^\ub G(s)ds+\phid-\phid G(\tau_1)\\
    &=\int_\lb^{\mu+\phid} G(s)ds-\int_{\lb}^{\tau_1}G(s)ds\\
&=\int_{\tau_1}^{\mu+\phid}G(s)ds\\
&=G(\tau_1)(\mu+\phid-\tau_1)\\
&=\frac{(\tau_b-\mu)(\mu-\tau_a)}{\tau_b-\tau_a}.
\end{align*}

Now we construct the following $G'$, which is in the class of step-exponential-step with the last step being degenerated:
$$G'(s)=\begin{cases}0&\text{ if }s\in[\lb,\tau_a)\\
e^\frac{\tau_a-\mu}{\tau_b-\mu}&\text{ if }s\in[\tau_a,\tau_a+\tau_b-\mu)\\
e^\frac{s-\tau_b}{\tau_b-\mu}&\text{ if }s\in[\tau_a+\tau_b-\mu,\tau_b]\\
1&\text{ if }s\in[\tau_b,\ub].
\end{cases}$$

It can be verified that $G'$ is a mean-preserving contraction of $G$, and $\tau_1'=\tau_a+\tau_b-\mu$ is a weak-highest equilibrium threshold under test-fee structure $(G',\phid',\phit')$ with $\phid'=\tau_b-\mu$ and $\phit'=0$. In this equilibrium, the intermediary's revenue is $(\tau_b-\mu)(1-e^{\frac{\tau_a-\mu}{\tau_b-\mu}})>\frac{(\tau_b-\mu)(\mu-\tau_a)}{\tau_b-\tau_a}$.

Now consider the other case $G(\mu+\phid)>G(\tau_1)$.  We construct a class of distributions $G_{\alpha}$ parametrized by $\alpha \in [0,1]$. 
    
    \begin{equation}\label{eq-step-exponential-step-improvement-1}
G_{\alpha}(s)=\begin{cases} \alpha G(s) &\text{ if }s \leq \thr_1 \\
\min\{\alpha G(\thr_1)e^{\frac{1}{\phid}(s - \thr_1)},G(\mu+\phid)\} & \text{ if } \thr_1 < s < \mu + \phid \\
G(s) &\text{ if } s \geq \mu + \phid.
\end{cases}
\end{equation}
 Notice that $G_{\alpha}$ is well-defined since $\phid > 0$, and is monotone non-decreasing and is between $0$ and $1$.  Therefore, $G_{\alpha}$ is a distribution.

We first show that there exists $\alpha \leq 1$ such that the areas under $G_{\alpha}$ and $G$ are equal.  We do so using a continuity argument.  Let $\thr_2(\alpha)$ be the lowest score $\thr' > \thr_1$ such that $G_{\alpha}(\thr') = G_{\alpha}(\mu+\phid)$. 
Consider $\alpha = 1$.  We have
\begin{align*}
    \int_\lb^{\thr_2(\alpha)} G_1(s)ds &= e^{\frac{1}{\phid}(\thr_2(\alpha)-\thr_1)}\int_\lb^{\thr_1} G_1(s)ds \\
    &= e^{\frac{1}{\phid}(\thr_2(\alpha)-\thr_1)}\int_\lb^{\thr_1} G(s)ds \\ 
    & \geq \int_\lb^{\thr_2(\alpha)} G(s)ds,
    \intertext{where the inequality followed from \autoref{lemma-boundspeed}.  As a result, since $G_{1}$ is weakly higher than $G$ for all scores above $\thr_2$, we have $\int_{\lb}^{\ub} G_1(s)ds \geq \int_{\lb}^{\ub} G(s)ds$.  Further,}
    \int_\lb^{\mu + \phid} G_0(s)ds &= 0 \leq \int_\lb^{\mu + \phid} G(s)ds.
\end{align*}
\noindent  As a result, since $G_0$ and $G$ are equal above $\mu + \phid$, we have $\int_{\lb}^{\ub} G_0(s)ds \leq \int_{\lb}^{\ub} G(s)ds$.  The area under $G_{\alpha}$ increases continuously in $\alpha$.  Therefore, there exists $\alpha \leq 1$ such that the areas under $G_{\alpha}$ and $G$ are equal.  For the rest of the proof fix such an $\alpha$.

We now show that $G_{\alpha}$ is a mean-preserving contraction of $G$.  Since $\alpha \leq 1$, the area under $G_{\alpha}$ up to any threshold $\thr' \leq \thr$ is weakly less than that for $G$.  For $\thr' \geq \thr_2(\alpha)$, since the total area under $G$ and $G_{\alpha}$ are equal and $G_{\alpha}$ is weakly higher than $G$ above $\mu + \phid$, we must have $\int_\lb^{\thr'} G_{\alpha}(s)ds \leq \int_\lb^{\thr'} G(s)ds$.  Finally, for any $\thr'$ such that $\thr \leq \thr' \leq \thr_2(\alpha)$ we have
\begin{align*}
        \int_\lb^{\thr'} G_{\alpha}(s)ds &= e^{-\frac{1}{\phid}(\thr_2-\thr')}\int_\lb^{\thr_2} G_{\alpha}(s)ds \leq e^{-\frac{1}{\phid}(\thr_2-\thr')}\int_\lb^{\thr_2} G(s)ds \leq \int_\lb^{\thr'} G(s)ds,
\end{align*}
\noindent where the inequality followed from \autoref{lemma-boundspeed}.  Thus, $G_{\alpha}$ is a mean-preserving contraction of $G$.

Test-fee structure $(G_{\alpha},\phi)$ has weakly higher revenue than $(G,\phi)$.  Since the two distributions are equal above $\mu + \phid$, the \eqref{eq-phiT-relaxed-opt} constraint is satisfied for $(G_{\alpha},\phi)$.  Further, $\thr$ is a weak-highest equilibrium threshold for $(G_{\alpha},\phi)$.  Since $\alpha \leq 1$, the disclosure probability in $(G_{\alpha},\phi)$ is weakly higher than $(G,\phi)$.  In fact, if $\alpha < 1$, the robust revenue of $(G_{\alpha},\phi)$ is more than that of $(G,\phi)$.

To see the first statement of the lemma, consider a score distribution $G'$ that is equal to $G_{\alpha}$ except that it pools the scores below $\thr_1$ and also pools the scores above $\thr_2$.  Formally,
$$G'(s)=\begin{cases}G_{\alpha}(\thr_1)&\text{ if }s \in [\thr_0,\thr_1],\\
G_{\alpha}(s)&\text{ if }s\in[\thr_1,\thr_2],\\
1&\text{ if }s \in [\thr_3,\ub],
\end{cases}$$
\noindent where $\thr_0$ and $\thr_3$ are such that the areas under $G'$ and $G_{\alpha}$ are equal.  Notice that the test-fee structure $(G',\phi)$ is in the step-exponential-step class, with a non-degenerated exponential part because $\thr_2>\thr_1$.  Further distribution $G'$ is a mean-preserving contraction of $G_{\alpha}$ and therefore $G$.  Finally, the robust revenue of the test-fee structure $(G',\phi)$ with weak-highest equilibrium threshold $\tau_1$ is equal to that of $(G_{\alpha},\phi)$, and therefore at least that of $(G,\phi)$.  This establishes the first statement of the lemma. 

To see the second statement, notice that if $(G,\phi)$ is optimal, then by \autoref{proposition-arbitrary-prior-exponential}\ref{Bullet-StrictlyPositive}, $\phid > 0$.  Further, it must be that $\alpha = 1$, as otherwise the revenue of $(G_{\alpha},\phi)$ is strictly higher.  Therefore, $G = G_1$, which means that $G$ is exponential from $\thr_1$ to $\thr_2$, and flat from $\thr_2$ to $\mu +\phid$, as claimed. 
\end{proof}

\begin{proof}[Proof of \autoref{proposition-arbitrary-prior-exponential}\ref{Bullet-ExistenceSES}]
Consider any robustly optimal test-fee structure $(G,\phi)$.  By \autoref{proposition-arbitrary-prior-exponential}\ref{Bullet-StrictlyPositive}, $\phid > 0$.  By \autoref{lemma-exponential-middle}, the robust revenue of $(G,\phi)$ is at most the robust revenue of some step-exponential-step test-fee structure $(G',\phi)$ where $G'$ is a mean-preserving contraction of $G$.  Therefore, $(G',\phi)$ is robustly optimal.

Next we show that in fact there is a unique equilibrium under any robustly optimal step-exponential-step test-fee structure with a slightly lower disclosure fee. Let $(G,\phi)$ be a robustly optimal step-exponential-step test-fee structure: 
\begin{align*}
G(s)=\begin{cases}g &\text{ if }s \in  [\thr_0,\thr_1) \\
ge^{(s - \thr_1)/(\thr_1-\thr_0)} & \text{ if } s \in [\thr_1,\thr_2)\\
1 &\text{ if } s\geq \thr_3,
\end{cases},
 \end{align*}
the disclosure fee is $\phid = \thr_1 - \thr_0$, and the testing fee is $\phit = (1-ge^{(\thr_2 - \thr_1)/(\thr_1-\thr_0)})(\thr_3 - (\mu + \phid))$. Notice that for any $\tau\in[\tau_1,\tau_2]$, 
$$E_G[s|s\leq \tau]=\tau-\frac{\int_\lb^\tau G(s)ds}{G(\tau)}=\tau-\frac{g(\tau_1-\tau_0)e^{\frac{\tau-\tau_1}{\tau_1-\tau_0}}}{ge^{\frac{\tau-\tau_1}{\tau_1-\tau_0}}}=\tau-\tau_1-\tau_0=\tau-\phid.$$
For $\tau>\tau_2$, $E_G[s|s\leq \tau]\leq \tau-\phid$. For $\tau\in[\tau_0,\tau_1)$,  $E_G[s|s\leq \tau]=\tau_0<\tau-\phid$.

  Consider a test-fee structure $(G,\phi')$ with $\phit'=\phit$ and $\phid'=\phid-\varepsilon$ for $\varepsilon<\tau_1-\tau_0$, then there exists a unique threshold $\thr'=\thr_1-\varepsilon$ satisfying \eqref{eq:robustimplementation}. Moreover, since $(G,\phi)$ satisfies \eqref{eq-phiT-relaxed-opt}, $(G,\phi')$ satisfies \eqref{eq:2} because $\int_{\mu + \phid}^{\bar{\theta}} [s - (\mu+\phid)] dG<\int_{\mu + \phid-\varepsilon}^{\bar{\theta}} [s - (\mu+\phid-\varepsilon)] dG$. Therefore, in the unique equilibrium under $(G,\phi')$, the asset is tested with probability 1 and the agent discloses all the scores above $\thr_1$.
\end{proof}

  \subsubsection{Proof of \autoref{prop-binaryrobustlyoptimal} on p. \pageref{prop-binaryrobustlyoptimal}}
\begin{proof}
    With binary support, the mean-preserving contraction constraints become $E_G[s]=E_F[\theta]$.  Therefore, the proposition follows from \autoref{lem:the-relaxed-problem}.
\end{proof}

\subsubsection{Proof of \autoref{proposition-log-concave-testing-fee-positive} on p. \pageref{proposition-log-concave-testing-fee-positive}}
We first simplify the step-exponential-step distribution identified in \autoref{Section-RobustlyOptimalTestFeeStructures} in the case where the testing fee is zero.

\begin{lemma}\label{lemma-log-concave-testing-fee-positive-lem1}
	Consider an optimal step-exponential-step test-fee structure defined in \eqref{eq-G(x)-arbitary-prior}, \eqref{eq-G(x)-arbitary-prior-feeD}, and \eqref{eq-G(x)-arbitary-prior-feeT}.  If $\phit^* = 0$, then $\thr_2 = \mu + \phid^*$ and $G^*(\thr_2) = 1$.  Further, if the prior distribution is log-concave, then the mean-preserving constraints are slack, $\int_\lb^{\thr'} G^*(s)ds < \int_\lb^{\thr'} F(s)ds$, for all $\thr'$ from $\thr_1$ to $\ub$.
\end{lemma}
\begin{proof}
	For the optimal testing fee to be zero, the area above $G^*$ from $\mu + \phid^*$ to $\ub$ must be zero, and hence $G^*(\mu + \phid) = 1$. From \autoref{lemma-exponential-middle}, $G^*(\tau_2)=1$.

	The constraint that the area under the score distribution $G^*$ is $\ub-\mu$ can be written as
	\begin{align*}
	\ub-\mu &= \int_\lb^{\thr_1} G^*(s)ds + \int_{\thr_1}^{\thr_2} ge^{\frac{s-\thr_1}{\thr_1 - \thr_0}}ds + (\ub-\thr_2) \\
	&= g(\thr_1 - \thr_0) +  (\thr_1 - \thr_0)ge^{\frac{s-\thr_1}{\thr_1 - \thr_0}}|_{\thr_1}^{\thr_2} + (\ub-\thr_2) \\
	&= (\thr_1 - \thr_0) + \ub-\thr_2.
	\end{align*}
	\noindent We must therefore have $\thr_2 = \mu + (\thr_1 - \thr_0) = \mu + \phid^*$.
	
	We now show that the MPC constraints must be slack on interval $[\thr_1,1]$.  Notice that the constraints are slack on $[\thr_2,\ub]$ since $G$ is equal to $1$ but $F(\thr')$ is strictly less than one for all $\thr' < 1$.  It thus remains to show that the constraints are slack on some interval $[\thr'_1,\thr_2]$. Notice that log-concavity implies continuity in the interior, which will be used in the later arguments.

	We first claim that $F(\thr_1)\geq G(\thr_1)=e^{-\frac{\mu-\thr_0}{\thr_1-\thr_0}}$. Suppose $F(\thr_1)< G(\thr_1)$, then $F(x)<G(x)$ for any $x\in[\thr_0,\thr_1)$, which implies the constraint must be slack on the interval because $\int_\lb^xF(s)ds=\int_\lb^{\thr_1} F(s)ds-\int_x^{\thr_1} F(s)ds>\int_\lb^{\thr_1} G(s)ds-\int_x^{\thr_1} G(s)ds=\int_\lb^xG(s)ds$. From the right-continuity of $F$, the constraint is also slack at $\thr_1$. If this is the case, we can construct a new distribution parameterized by $\thr_1'=\thr_1-\varepsilon$ and $\thr_0'=\thr_0-\varepsilon$, the distribution on $[\thr_1,\ub]$ doesn't change so all the constraints on $[\thr_1,\ub]$ are still satisfied. Also the original constraints on $[\lb,\thr_1]$ are slack, they are still satisfied for small $\varepsilon$. By charging the same disclosure fee $\phid=\thr_1-\thr_0$, this new distribution can induce a higher disclosure probability, which contradicts to $G$ being optimal.
	
	If $F(\thr_1)=G(\thr_1)$, the same argument goes through if the constructed distribution doesn't violate the mean-preserving constraints. If the constraint is slack at $\thr_1$, then all the constraints to the left are slack, so the constructed distribution is still a profitable deviation. If the constraint binds at $\thr_1$, the right derivative of $F$ at $\thr_1$ must be greater than the right derivative of $G$. Also the left derivative of $F$ must be greater than $G$ from log-concavity. So the local change of the distribution doesn't violate any constraint either, which leads to a profitable deviation.
	
	So we must have $F(\thr_1)>G(\thr_1)$, and the constraint at $\thr_1$ is slack at $\thr_1$ due to the continuity of $F$. Moreover, from the log-concavity of $F$, and $F(\thr_1)>G(\thr_1)$, $F(\thr_2)<G(\thr_2)$, $F$ crosses the exponential part of $G$ from above once and only once. To see this, notice that $\log(F)$ is concave and $\log(G)$ is linear on $[\thr_1,\thr_2]$, and a concave function can only cross a linear function from above once. Let $x^*$ denote the intersection point, we have $F(x)>G(x)$ for $x\in[\thr_1,x^*)$ and $F(x)<G(x)$ for $x\in(x^*,\thr_2]$. Notice that the constraint is slack at $\thr_2$, which means $\int_\lb^{\thr_2} G(s)ds\leq \int_\lb^{\thr_2} F(s)ds$, then for any $x\in[x^*,\thr_2)$, the constraint is also slack because $\int_\lb^x G(s)ds=\int_\lb^{\thr_2} G(s)ds-\int_x^{\thr_2} G(s)ds< \int_\lb^{\thr_2} GFds-\int_x^{\thr_2} F(s)ds=\int_\lb^x F(s)ds$. Similarly the constraint is slack at $\thr_1$ implies the constraint is also slack at any $x\in(\thr_1,x^*]$.
	
	Therefore, all MPC constraints are slack on $[\thr_1,\ub)$.
\end{proof}

Given \autoref{lemma-log-concave-testing-fee-positive-lem1} we now prove \autoref{proposition-log-concave-testing-fee-positive}.

\begin{proof}[Proof of \autoref{proposition-log-concave-testing-fee-positive}]
	Consider the class of step-exponential-step distributions defined in \eqref{eq-G(x)-arbitary-prior} where $\thr_3 = \ub$.  That is,
	\begin{equation*}
	G(s)=\begin{cases} g &\text{ if }s = \thr_0 \\
	ge^{(s - \thr_1)/(\thr_1-\thr_0)} & \text{ if } s \in [\thr_1,\thr_2]\\
	1 &\text{ if } s= \ub.
	\end{cases}
	\end{equation*}
	\noindent A degree of freedom is removed given the constraint that the area under $G$ is $\ub-\mu$.  Notice that the distribution with zero testing fee identified in \autoref{lemma-log-concave-testing-fee-positive-lem1} is a special case where $\thr_2 = \mu + (\thr_1 - \thr_0)$.  Consider such a distribution.  We show that reducing $\thr_1$ by a small amount increases the revenue, and does not violate any of the constraints.
	
	The intermediary's revenue from charging disclosure fee $\phid = \thr_1 - \thr_0$ and the highest possible testing fee $\phit=\int_{\mu+\phid}^\ub (s-\mu-c)dG(s)$ is
	\begin{align*}
	\phit+\phid(1-G(\thr_1))
	=&\int_{\mu+x_1-x_0}^\ub [s-\mu-(\thr_1-\thr_0)]dG(s)+(\thr_1-\thr_0)(1-G(\thr_1))\\
	=&\ub-\mu-(\thr_1-\thr_0)-\int_{\mu+\thr_1-\thr_0}^\ub G(s)ds+\thr_1-\thr_0-\int_\lb^{\thr_1}G(s)ds\\
	=&\int_{\thr_1}^{\mu+\thr_1-\thr_0}G(s)ds\\
	=&ge^{\frac{\thr_2-\thr_1}{\thr_1-\thr_0}}(2(\thr_1-\thr_0)+\mu-\thr_2)-(\thr_1-\thr_0)g.
	\end{align*}

	Consider the change in revenue from decreasing $\thr_1$, evaluated at $\thr_2 = \mu + (\thr_1 - \thr_0)$ and $g = e^{-\frac{\mu - \thr_0}{\thr_1 - \thr_0}}$,
	\begin{align*}
	\frac{d}{d\thr_1}(\phit+\phid(1-G(\thr_1)))=&(\thr_1-\thr_0)\frac{d}{d\thr_1}\left(\frac{\thr_2-\thr_1}{\thr_1-\thr_0}\right)+2-\frac{d\thr_2}{d\thr_1}-g\\
	=&\left(\frac{d\thr_2}{d\thr_1}-1\right)-\frac{\thr_2-\thr_1}{\thr_1-\thr_0}+2-\frac{d\thr_2}{d\thr_1}-g\\
	=&1-g-\frac{\thr_2-\thr_1}{\thr_1-\thr_0}\\
	=&1+\frac{\thr_0-\mu}{\thr_1-\thr_0}-e^{\frac{\thr_0-\mu}{\thr_1-\thr_0}}\\
	<&0.
	\end{align*}
	The inequality follows from $e^x > 1 + x$ for $x> 0$. Thus decreasing $\thr_1$ strictly increases revenue. 
\end{proof}

\subsection{Proofs for \autoref{Section-Extension}}
\subsubsection{Proof of \autoref{prop-evidence} on p. \pageref{prop-evidence}}

\begin{proof}

Consider an evidence-test-fee structure, denoted by fees $(\phit,\phid)$, an unbiased test $T:\Theta\rightarrow\Delta S$, and  an evidence structure
$M: S\rightrightarrows \mathcal{M}$ such that for each $s$, $M(s)$ is a Borel space. A strategy profile $(\sigma,p)$ consists of the agent's strategy $\sigma=(\sigma_T,\sigma_D)$, where $\sigma_T\in [0,1]$ and $\sigma_D$ maps $s\in S$ to $\Delta(M(s)\cup \{N\})$, and the market price $p:\mathcal{M}\rightarrow [\underline{\theta},\bar{\theta}]$.
Let $(\sigma,p)$ be an adversarial equilibrium. First we consider the case in which the agent has the asset tested with probability $1$, that is $\sigma_T=1$.

Consider the disclosure stage.  Let $G_{(\sigma,p)}$ be the induced distribution of prices, i.e., $G_{(\sigma,p)}(x) = Pr[p(\sigma_D(s)) \leq x]$ for any $x$, taking into account both the randomization over the score and the agent's strategy.  Let ${\tau} \equiv p(N) + \phid$.  We show that the following holds, mirroring our characterization of highest equilibria \eqref{eq:robustimplementation}:
\begin{align}
&\tau - \phid \leq E_{G_{(\sigma,p)}}[x | x \leq \tau],\label{eq1}\\
&\tau' - \phid > E_{G_{(\sigma,p)}}[x | x \leq \tau'], \forall \tau' > \tau\label{eq2}.
\end{align} 
Let us argue why \eqref{eq1} holds. Since $\tau - \phid = p(N)$, it suffices to show that $p(N)$ is weakly less than $E_{G_{(\sigma,p)}}[x | x \leq \tau]$. Observe that $p(\sigma_D(s))$ is at least $p(N)$ with probability $1$: if $p(\sigma_D(s))$ were strictly less than $p(N)$, the agent could profitably deviate to sending message $N$ and obtain a strictly higher price. But then this implies that $p(N) \leq E[p(\sigma_D(s)) | p(\sigma_D(s)) \leq \tau]= E_{G_{(\sigma,p)}}[x | x \leq \tau]$. 

To see that \eqref{eq2} holds, suppose for contradiction that $\tau'' - \phid \leq E_{G_{(\sigma,p)}}[p | p \leq \tau'']$ for some $\tau'' > \tau$.  By \autoref{lemma-inter} there exists $\tau' > \tau$ such that $\tau' - \phid = E_{G_{(\sigma,p)}}[p | p \leq \tau']$.  Consider the strategy profile $(\sigma',p')$ defined as follows.  The agent's strategy $\sigma'$ is the same as $\sigma$ except that the agent conceals a score $s$ if $p(\sigma(s)) \leq \tau'$.  For $m \neq N$, $p'(m) = p(m)$, and for $m = N$, the price is $p'(m) = E_{G_{(\sigma,p)}}[x | x \leq \tau'] \geq E_{G_{(\sigma,p)}}[x | x \leq \tau] \geq p(m)$.  Notice that since any message $m$ that is disclosed in $(\sigma',p')$ is also disclosed in $(\sigma,p)$, the prices $p'$ are defined on path via the Bayes rule.

To see that $(\sigma',p')$ is an equilibrium, consider any score $s$ such that $p(\sigma(s)) > \tau$ with positive probability.  Since $s$ can send a message $m \neq N$ such that $p(m) > \tau$, it optimally randomizes over messages other than $N$ that lead to the same (and maximum) price, which, abusing notation, we denote $p(\sigma(s))$. Since $p(m) = p'(m)$ for all $m \neq N$,  $\sigma(s)$ is optimal among all strategies that send $N$ with probability 0 given prices $p'$.  Therefore, for such a score, it is optimal to follow $\sigma(s)$ if $p'(\sigma(s)) > \tau' = p'(N) + \phid$, and  to conceal if $p'(\sigma(s)) \leq \tau'$, as prescribed by $\sigma'$.  Now consider a score $s$ such that $p(\sigma(s)) \leq\tau$ with probability 1.  For such a score, it is optimal given prices $p$ to conceal, i.e., for any message $m \neq N$ that $s$ can send, $p(m) - \phid \leq p(N)$.  Since $p'(m) = p(m)$ and $p'(N) \geq p(N)$, it is also optimal to conceal given prices $p'$, as prescribed by $\sigma'$.

Now consider the testing stage.  If
\begin{align}\label{Equation-ParticipationViolated2}
\phit \geq \int_{\mu+\phid}^{\overline\theta}[x-(\mu+\phid)]dG_{(\sigma,p)}.
\end{align}
\noindent then there exists an equilibrium in the evidence-test-fee structure where the agent has the asset tested with probability $0$. The argument parallels that of \autoref{lemma-worker-takes-test-wp1-in-all}.  Consider a strategy profile $(\sigma',p')$ such that that $\sigma'_T = 0$, off-path the agent follows $\sigma(s)$ if $p'(\sigma(s)) > \mu + \phid$ ($p'(\sigma(s))$ is well-defined as argued above) and otherwise conceals, and the prices are $p'(N) = \mu$, and $p(m) = p'(m)$ for all $m \neq N$.  Since the set of disclosed messages in $(\sigma',p')$ is a subset of that in $(\sigma,p)$, an argument similar to above shows that the agent's disclosure strategy is sequentially rational. Also, by deviating to take the test, the agent receives an expected payoff lower than $\mu$,
\begin{align*}
\int_{\underline{\theta}}^{\mu+\phid}\mu dG_{(\sigma,p)}+\int_{\mu+\phid}^{\bar{\theta}} [x-\phid]dG_{(\sigma,p)}-\phit&=\mu+\int_{\mu+\phid}^{\bar{\theta}}[x-(\mu+\phid)]dG_{(\sigma,p)}-\phit\leq \mu
  \end{align*}
where the equality is algebra, and the inequality follows from \eqref{Equation-ParticipationViolated2}.  Therefore, adversarial revenue is at most zero, which is obtained by any test-fee structure with zero fees.  So suppose that \eqref{Equation-ParticipationViolated2}  is violated.

Now consider a test-fee structure $(G_{(\sigma,p)},\phi)$.  By \autoref{lemma-highest-threshold}, $\tau$ is a weak-highest equilibrium threshold of the test-fee structure.  Also since \eqref{Equation-ParticipationViolated2}  is violated, by \autoref{lemma-worker-takes-test-wp1-in-all}, the test is taken with probability 1 in all equilibria.  Therefore, adversarial revenue of the test-fee structure is equal to the adversarial revenue of the evidence-test-fee environment.

Now we consider the case $\sigma_T\in[0,1)$. Notice that in this case 
$$\phit\geq \int_{p(N)+\phid}^\ub[x-(p(N)+\phid)] dG_{(\sigma,p)}\geq \int_{\mu+\phid}^{\overline\theta}[x-(\mu+\phid)]dG_{(\sigma,p)}$$
so \eqref{Equation-ParticipationViolated2} holds. Form the same argument above we know that the adversarial revenue is at most zero, which can be obtained in a test-fee structure with any test and zero fees.
\end{proof}

\end{document}